\documentclass[%
 reprint,
superscriptaddress,
nofootinbib,
 amsmath,amssymb,
 aps,
pra,
]{revtex4-2}

\usepackage{graphicx}
\usepackage{dcolumn}
\usepackage{bm}
\usepackage{hyperref}

\usepackage{amsmath}
\usepackage{amssymb}
\usepackage{bm}
\usepackage{stackrel}
\usepackage{tikz-cd}
\usepackage{tikz}
\usepackage[cm]{fullpage}
\usepackage{amsthm}
\usepackage{mathtools}
\usepackage{enumitem}
\usepackage[normalem]{ulem}
\usepackage{hyperref}
\usepackage[capitalize]{cleveref}
\usepackage{appendix}
\usepackage{comment}
\usepackage{subcaption}
\usepackage{centernot}
\usepackage{braket}
\usepackage{algorithm}
\usepackage{algpseudocode}
\usepackage{varwidth}
\usepackage{ulem}
\usepackage{apptools}
\AtAppendix{\counterwithin{Lemma}{section}}
\AtAppendix{\counterwithin{Theorem}{section}}
\AtAppendix{\counterwithin{Proposition}{section}}
\AtAppendix{\counterwithin{Corollary}{section}}

\newtheorem{Theorem}{Theorem}
\newtheorem{Observation}{Observation}
\newtheorem{Lemma}{Lemma}
\newtheorem{Proposition}{Proposition}
\newtheorem{Corollary}{Corollary}

\newtheorem{Result}{Result}

\theoremstyle{definition}

\begin{document}

\preprint{APS/123-QED}

\title{Necessary and Sufficient Conditions for Universal Gates with Pauli Strings and Beyond}

\author{Isaac D. Smith}
\affiliation{ICFO-Institut de Ciencies Fotoniques, The Barcelona Institute of Science and Technology, Av. Carl Friedrich Gauss, 08860 Castelldefels, Barcelona, Spain}

\author{Hans J. Briegel}
\affiliation{University of Innsbruck, Department of Theoretical Physics, Technikerstr. 21A, Innsbruck A-6020, Austria}

\author{Hendrik Poulsen Nautrup}
\email{hendrik.poulsen-nautrup@uibk.ac.at}
\affiliation{University of Innsbruck, Department of Theoretical Physics, Technikerstr. 21A, Innsbruck A-6020, Austria}

\date{\today}

\begin{abstract} 
Any quantum computation can be represented as a sequence of unitary evolutions generated by a finite set of Hamiltonians. For the case where this set consists of only products of Pauli operators, known as Pauli strings, we provide a necessary and sufficient condition for it to generate $\mathfrak{su}(2^n)$, i.e., to be universal for quantum computation on $n$ qubits. When combining Pauli strings with a general Hamiltonian, we show a sufficient (and in certain circumstances even necessary) condition for universality based on the Pauli-basis expansion of the Hamiltonian. As an application of these results, we prove two corollaries: (i) a necessary and sufficient condition for the universality of a general Hamiltonian given arbitrary single-qubit control on all qubits, and (ii) the universality of an XYZ Heisenberg Hamiltonian with local control of just two adjacent qubits.
\end{abstract}

\maketitle

\section{Introduction}

Suppose you have a quantum system at hand, and a number of ways of controlling said system. A fundamental question to ask is, what exactly can this system be used for? In the context of quantum computation, this question can be rephrased as, which {\em computations} can this system be made to perform?

From the perspective of quantum control theory \cite{huang1983controllability,d2007introduction,albertini2001notions,albertini2002lie,albertini2021subspace,Ramakrishna,Schirmer_01}, which has substantial overlap with both theoretical and experimental approaches to quantum computing \cite{Ram_96}, the above question is typically phrased in terms of the Hamiltonians describing the controllable evolutions of the system. As the time evolution generated by a Hamiltonian corresponds to unitary evolution on the state space of the quantum system, the above question can be further refined to, which {\em unitaries} can be implemented by the set of Hamiltonian controls that I have?

If, for a given system and set of Hamiltonians, the answer to the above question is ``all of them'', then the system and Hamiltonians are considered to be {\em universal} \cite{d2007introduction,albertini2001notions}. In the context of quantum computing, this translates to being a general purpose computer able to perform any possible computation and is thus a desirable property. However, demonstrating that a set of Hamiltonians suffices for universality is typically a non-trivial endeavor. 

In this work, we establish several criteria that are necessary and sufficient for certain sets of Hamiltonians to satisfy in order to be universal. These criteria are, moreover, readily checkable and are of practical use in several realistic settings (see, e.g., the corollaries presented in \Cref{sec:formal_results}). Specifically, we consider quantum systems consisting of $n$ qubits and sets of Hamiltonians containing Pauli strings, that is, $n$-fold tensor products of single-qubit Pauli operators \cite{nielsen2010quantum}. Concretely, we establish several sufficient and also necessary and sufficient criteria for universality for sets of Hamiltonians comprised of: (i) Pauli strings only, (ii) Pauli strings and a single other Hamiltonian, and (iii) Hamiltonians with local control on differing numbers of qubits. These results are stated informally in the following subsection.

Why focus on Pauli strings? There are several practical and pragmatic reasons for doing so. On the practical side, Pauli strings naturally appear in various contexts of theoretical and experimental quantum computation, for example in measurement-based quantum computation \cite{Raussendorf2003,briegel2009measurement} and in quantum computation using trapped ions \cite{debnath2016demonstration}. Pragmatically, the set of Pauli strings exhibit several favorable mathematical properties: they form a (real) basis for the set of all $n$-qubit Hermitian operators, and any pair of them either commute or anti-commute. These latter properties allow for analytical progress to be made on the question of demonstrating universality and form the foundation for the criteria established in this work. For these and other reasons the study of Pauli strings has received significant attention recently (see, e.g., Refs~\cite{aguilar2024classificationpauliliealgebras,smith2025optimally,cuypers2026dynamicallie,gargiulo2026paulistringsquantumdynamics}).

\subsection{Results} \label{subsec:results}

In quantum control theory in general, and in this work in particular, methods based on Lie algebra theory \cite{hall2013lie,knapp2002liegroups} are commonly used \cite{albertini2002lie,jurdjevic1972control,schirmer2002identification}. However, as these methods are often technical and since one of the aims of this work is to provide readily accessible and applicable criteria for checking universality, we outline the key results in an informal manner here, leaving the corresponding formal statements to \Cref{sec:formal_results}. 

A core guiding principle in this work, stems from the following observation:

\textbf{Observation:} \textit{A set of Hamiltonians $\mathbb{H}$ is a universal generating set if and only if it can generate any other universal generating set $\mathbb{H}'$}.

Accordingly, while it may be difficult to demonstrate directly that a set of Hamiltonians of interest is a universal generating set, it may be possible to reduce to problem to showing that a simpler set of Hamiltonians is universal.

In this spirit, the simpler sets of Hamiltonians that we consider here will typically be sets comprising entirely of Pauli strings. Pauli strings exhibit several favorable features: the set of all such strings form a (real) basis for the space of $n$-qubit Hamiltonians, they can be mapped to vector spaces over binary fields \cite{nielsen2010quantum}, and each pair either commutes or anti-commutes. Using this latter fact, it is possible to consider a graph where each vertex is associated to a Pauli string in the set under consideration and where an edge exists between vertices if the corresponding Pauli strings anti-commute (see \Cref{sec:formal_results} for a formal definition). Our first result can then be (informally) stated as follows:
\begin{Result}[See \Cref{thm:nec_suf_pauli}] \label{res:Pauli_strings_only} A set of $n$-qubit Pauli strings $\mathbb{P}$ ($n \ge 2$) is universal if and only if 
\begin{enumerate}[label=(\roman*)]
    \item the set of Pauli strings obtained from $\mathbb{P}$ by taking iterated commutators contains a subset that is universal on $2\le k < n$ qubits,
    \item taking products of elements of $\mathbb{P}$ produces all Pauli strings on $n$ qubits, and
    \item the anti-commutation graph of $\mathbb{P}$ is connected.
\end{enumerate}
\end{Result}
The second and third criteria can be directly checked from the set $\mathbb{P}$ under consideration, with the latter being a property of the anti-commutation graph. The first criterion also aligns with the guiding principle above: while it can require some effort to establish in general, in many cases it is straightforward, such as in the simplest case where $k=2$. Furthermore, recent results from Ref.~\cite{cuypers2026dynamicallie} may provide efficient methods for checking this criterion from the anti-commutation graph as well (see also Ref.~\cite{gargiulo2026paulistringsquantumdynamics}).

Now suppose that, instead of only having Pauli strings available, we have a set of Pauli strings $\mathbb{Q}$ and also a more general Hamiltonian $H$. It is now no longer the case that all elements of the set of control Hamiltonians either commute or anti-commute, so the notion of anti-commutation graph does not directly extend to this case. Nevertheless, it is possible to associate a graph to the joint set of Pauli strings and $H$ in a productive way. 

As stated above, the set of Pauli strings form a real basis for the set of Hamiltonians, and so $H$ can be written as a real linear combination of Pauli strings. Accordingly, it is possible to consider the anti-commutation graph of the set of Pauli strings comprised of $\mathbb{Q}$ as well as those appearing (with a non-zero coefficient) in the decomposition of $H$. Furthermore, through suitable application of the controls in $\mathbb{Q}$ and of $H$, it is possible to, in effect, enact an {\em overall} control corresponding to a Pauli string in the latter set. {\em Which} Pauli strings fall into this category can be determined using a graphical operation - called a {\em unique neighbor expansion} below, see \Cref{sec:formal_results} - acting on the anti-commutation graph built from $\mathbb{Q}$ and $H$. This is the content of the next result:
\begin{Result}[See \Cref{lem:isolation}] \label{res:isolation} The unitary $e^{-itP}$ generated by a Pauli string $P$ can be implemented if $P$ corresponds to a vertex in the (iterated) unique neighbor expansion of the anti-commutation graph of $\mathbb{Q}$ and $H$.
\end{Result}
While the exact definition of the unique neighbor expansion operation is postponed until \Cref{sec:formal_results}, there are two key points related to the above result worth highlighting here. First, since the result is again phrased in graphical terms, it remains a directly checkable criterion as per the aim of producing accessible universality criteria that is central to this work. Second, it can allow a question of universality of Pauli strings and a general Hamiltonian to be reduced to a question of universality of the set of Pauli strings produced by unique neighbor expansion, to which the first result above can be applied. That is,
\begin{Result}[See \Cref{thm:suf_ham}] \label{res:universality_H} If a set of Pauli strings $\mathbb{P}$ produced by (iterated) unique neighbor expansion of the set $\mathbb{Q}$ and a Hamiltonian $H$ is universal, then $\mathbb{Q} \cup \{H\}$ is universal.
\end{Result}
In certain circumstances, the above sufficiency statement is also necessary.

Finally, we demonstrate the utility of the above results by considering two scenarios of interest. The first consists in a set of general Hamiltonian controls $\mathbb{H}$ (with no constraints on being Pauli strings) that contains a subset $\mathbb{H}_{\text{s.q.}}$ that can implement any single qubit unitary operation on any of the $n$ qubits but no entangling operations, and a single additional Hamiltonian $H$ which can produce entanglement. That is, all local control is possible and a single choice of non-local control is available. In such a scenario, we have the following result:
\begin{Result}[See \Cref{cor:local_control}] \label{res:local_control} The set of Hamiltonians $\mathbb{H}_{\text{s.q.}} \cup \{H\}$ is universal if and only if the set of Pauli strings appearing in the decomposition of $H$ with non-zero coefficients:
\begin{enumerate}[label=(\roman*)]
    \item contains at least one element with even support, and
    \item has a connected anti-commutation graph.
\end{enumerate}
\end{Result}
The second scenario of interest consists in a relaxation on the requirements on $\mathbb{H}_{\text{s.q.}}$ in the above scenario, but for a specific, yet practically relevant, choice of $H$. The relaxation on $\mathbb{H}_{\text{s.q.}}$ is that any single qubit unitary ( i.e., arbitrary local control) need only be implementable on a subset of qubits, while the choice of $H$ is the nearest-neighbor anisotropic Heisenberg, or $XYZ$, spin-chain Hamiltonian, a standard exchange-interaction model in quantum magnetism with historical origins in the theory of ferromagnetism \cite{heisenberg1928ferromagnetismus,auerbach1994quantummagnetism}, i.e.,
\begin{align}
    H = \sum_{j=1}^{n-1} J_{x}X_{j}X_{j+1} + J_{y}Y_{j}Y_{j+1} + J_{z}Z_{j}Z_{j+1},
\end{align}
where the scalar values $J_{x}$, $J_{y}$ and $J_{z}$ are real and non-zero, and where $X$, $Y$ and $Z$ denote the Pauli operators. In this case, we have the following:
\begin{Result}[See \Cref{cor:heisenberg}] \label{res:Heisenberg} Suppose that $\mathbb{H}_{\text{s.q.}}$ generates arbitrary local control on the first $k$ qubits and let $H$ be as above. Then $\mathbb{H}_{\text{s.q.}} \cup \{H\}$ is universal if and only if $k\geq2$. 
\end{Result}

\subsection{Discussion and outlook}\label{sec:discussion}

This work gives a broad perspective on the universality of dynamical Lie algebras generated by multi-qubit Pauli strings as well as more general Hamiltonians. On the one hand, \Cref{res:Pauli_strings_only} gives a necessary and sufficient criterion for universality in the case where all Hamiltonians are Pauli-strings. On the other hand, \Cref{res:isolation} and \Cref{res:universality_H} show that, in some cases, questions of universality for more general Hamiltonians can be reduced to the former via the isolation of Pauli terms in the Hamiltonian expansion. From this perspective, \Cref{res:local_control} and \Cref{res:Heisenberg} are not isolated corollaries, but rather instances of the same underlying mechanism. In particular, previously known universality statements such as these, (they were discussed in Ref.~\cite{bremner2004fungible} and Ref.~\cite{zeier2011symmetry}, respectively) arise here in a unified framework.

As the case of Pauli strings in \Cref{res:Pauli_strings_only} forms the basis for all the universality results presented here, it should also be viewed in the context of recent, parallel developments. In Ref.~\cite{aguilar2024classificationpauliliealgebras}, a characterization of Lie algebras was given in terms of a standard form defined with respect to a notion of contraction of the anti-commutation graph of the corresponding set of Pauli strings. This method can be used to determine whether a set of Pauli strings $\mathbb{P}$ generates the full Lie algebra $\mathfrak{su}(2^n)$ or a sub-algebra thereof, but requires the set $\mathbb{P}$ to be minimal and also requires some amount of processing to produce the standard form of the anti-commutation graph. In Ref.~\cite{smith2025optimally}, only those sets $\mathbb{P}$ that generate the full Lie algebra $\mathfrak{su}(2^n)$ were considered, with a focus on demonstrating tight lower bounds on the size of universal sets $\mathbb{P}$ as well as on the efficiency of such sets for generating certain classes of unitary dynamics. Several results of this latter work are precursors to the results established here.

More recently, the setting of \Cref{res:Pauli_strings_only} has been treated in Ref.~\cite{cuypers2026dynamicallie}, which establishes a complete and graphical characterization of all dynamical Lie (sub)-algebras generated by Pauli strings using the anti-commutation graph of the set $\mathbb{P}$ directly. In particular, Ref.~\cite{cuypers2026dynamicallie} establishes, among other results, that a set $\mathbb{P}$ generates $\mathfrak{su}(2^n)$ if and only if  (i) its anti-commutation graph contains as a subgraph one of a finite set of ``forbidden'' graphs on $6$ vertices, as well as two additional criteria similar to (ii) and (iii) in \cref{res:Pauli_strings_only}. These results are established using tools arising from the study of quadratic spaces over binary fields and also lead to efficient algorithmic methods for checking universality.

Ultimately, however, sets of Pauli strings represent a restricted class of all possible sets of Hamiltonians generating unitary dynamics. It is thus prudent to investigate whether and how the results and techniques established in the Pauli-string setting generalize to other Hamiltonians. Here, we use the fact that any $n$-qubit Hamiltonian can be expanded in the Pauli string basis as a means for doing so. In this sense, the present results connect the recent theory of the universality of Pauli string generating sets to more general questions in Hamiltonian simulation and quantum control.

Conceptually, if a generating set consists of both Pauli strings and more general Hamiltonians, it may be possible to use the former to ``isolate'' the Pauli strings that appear in the expansions of the latter in the Pauli string basis. These isolated Pauli strings, along with the original Pauli strings included in the generating set, form a set of Pauli strings to which the results of Pauli string universality can be applied. In essence, this is the content of \Cref{res:isolation} and \Cref{res:universality_H}.

However, as it stands, Result~\ref{res:universality_H} is a sufficient criterion for universality but not a necessary one (in general). It essentially only relies on commutation and neglects linear combinations, and so it cannot be a necessary one without additional modifications. It remains, therefore, an open question how to extend the present criterion to capture all Pauli strings that can in principle be isolated. 

Let us elaborate. Given Hamiltonians $H_1,\dots,H_m$, one would like criteria, stated directly in terms of their Pauli expansions, that determine which Pauli string terms can be isolated so that the above Pauli string results can be applied. To see why the criteria related to \Cref{res:isolation}, that is, in terms of unique neighbor expansion, is insufficient in general, consider the following two examples. First, consider the following 1-qubit example:
\begin{align*}
iH_1^{(1)}&=\alpha\, iX,\\
iH_2^{(1)}&=\beta_1\, iY+\beta_2\, iZ,
\end{align*}
with $\alpha,\beta_1,\beta_2\in\mathbb{R}^*$. Although \Cref{res:isolation} suggests that no term in $iH_2^{(1)}$ can be isolated, in fact, all terms can be isolated by making use of linear combinations. This demonstrates that any criteria satisfying the above desiderata will likely need to involve more than just the anti-commutation graph of the set of all Pauli strings appearing in the set of Hamiltonians.

As a second example, consider the following 2-qubit example:
\begin{align*}
iH_1^{(2)}&=\alpha\, iIX,\\
iH_2^{(2)}&=\beta_1\, iIY+\beta_2\, iXZ+\beta_3\, iYZ+\beta_4\, iZZ,
\end{align*}
with $\alpha,\beta_1,\ldots,\beta_4\in\mathbb{R}^*$. In this case, although the set $\{iIX,iIY,iXZ,iYZ,iZZ\}$ arising as the set of all Pauli-strings appearing in the Hamiltonians above is universal on two qubits, the pair $\{iH_1^{(2)},iH_2^{(2)}\}$ is not (which is proven in \cref{app:counterexample}). Thus, any sharper criterion must retain information about how Pauli terms are grouped into the available control Hamiltonians, and make use of linear combinations. It would be interesting to identify a refinement of the present graph-theoretic picture that captures precisely this additional structure, and see how this refinement can be combined with the results presented here and also those presented in, e.g., Refs.~\cite{cuypers2026dynamicallie,gargiulo2026paulistringsquantumdynamics}.

Finally, it is natural to ask to what extent the present picture extends beyond qubits. For the setting of \Cref{res:local_control}, such an extension is already suggested by the results of Ref.~\cite{bremner2005simulating}, which show that, with arbitrary local unitary control, the obstruction associated with odd-body couplings is specific to the qubit case. This suggests that the parity condition appearing in the present qubit analysis is genuinely a two-level phenomenon. Extending the graph-theoretic and isolation-based methods developed here to generalized Pauli/Weyl operators on qudits, and to restricted local control in the spirit of \Cref{res:Heisenberg}, appears to be a promising direction for future work.

\section{Background and Definitions} \label{sec:background}

In the remainder of this manuscript, we provide the required background content for stating and proving the results from \Cref{subsec:results} in a formal manner. The present section is devoted to the former, while the next section, \Cref{sec:formal_results}, deals with the latter. Concretely, we provide here the required notation and background definitions from Lie algebras, the study of Pauli strings within quantum information theory, and graph theory relevant for the results and proofs in \Cref{sec:formal_results}.

In this work, we exclusively deal with quantum systems consisting of $n$ qubits, with state space given by $\mathcal{H} \cong \bigotimes_{j=1}^{n} \mathbb{C}^{2}$. The set of anti-Hermitian operators acting on $\mathcal{H}$ forms a real vector space and, for brevity, we will identify our Lie algebra of interest, denoted $\mathfrak{su}(2^{n})$, with its defining representation acting on $\mathcal{H}$. That is, we define the {\bf real Lie algebra $\mathfrak{su}(2^{n})$} to be the $(4^{n}-1)$-dimensional real vector space of traceless, skew-Hermitian operators equipped with a Lie bracket $[\cdot, \cdot]: \mathfrak{su}(2^{n}) \times \mathfrak{su}(2^{n}) \rightarrow \mathfrak{su}(2^{n})$ given by the matrix commutator, i.e., for all $iA, iB \in \mathfrak{su}(2^{n})$\footnote{A remark regarding notation. Since $\mathfrak{su}(2^{n})$ consists of {\em skew}-Hermitian operators, we will explicitly represent the imaginary unit $i$ where appropriate. This applies to individual operators, such as $iA$ and $iB$ considered here, as well as sets of operators, such as $i\mathbb{H}$. In the former case, $A$ and $B$ can be considered as Hermitian operators, and in the latter case, $\mathbb{H}$ can be considered as a set of Hermitian operators ($i\mathbb{H}$ is then understood to be the set obtained from $\mathbb{H}$ by applying $i$ to each element).},
\begin{align}
[iA, iB] := (iA)(iB) - (iB)(iA).
\end{align} 
For each $iA \in \mathfrak{su}(2^{n})$, the {\bf adjoint map} associated to $iA$ can be defined as $\textrm{ad}_{iA}: \mathfrak{su}(2^{n}) \rightarrow \mathfrak{su}(2^{n})$ via $\textrm{ad}_{iA}(iB) := \frac{1}{2}[iA, iB]$. The factor of a half is included for convenience when $iA$ and $iB$ are Pauli strings (see below).

Throughout this work, the sets of Hamiltonians whose universality are under consideration will be denoted using math-bold font, such as $i\mathbb{H} \subseteq \mathfrak{su}(2^{n})$ for general skew-Hermitian operators and $i\mathbb{P}, i\mathbb{Q} \subseteq \mathfrak{su}(2^{n})$ for Pauli strings. For a given set $i\mathbb{H}$, an important related set is the set of all skew-Hermitian operators that result from iteratively applying the adjoint map to $i\mathbb{H}$. Let us define $i\mathbb{H}_{\textrm{ad}^{(0)}}:= i\mathbb{H}$ and, for all $r \ge 1$, 
\begin{align}
i\mathbb{H}_{\textrm{ad}^{(r)}} := \{ \textrm{ad}_{iA}(iB) | iA \in i\mathbb{H}, iB \in i\mathbb{H}_{\textrm{ad}^{(r-1)}} \}.
\end{align}
Let us also define
\begin{align}
\langle i\mathbb{H}\rangle_{[\cdot,\cdot]}:=\bigcup_{r=0}^{\infty} i\mathbb{H}_{\textrm{ad}^{(r)}}.
\end{align}
With this notation, the {\bf Lie closure} of $i\mathbb{H}$, denoted by $\langle i\mathbb{H} \rangle_{\textrm{Lie}}$, is the real vector space of linear combination of elements of $\langle i\mathbb{H}\rangle_{[\cdot,\cdot]}$, i.e.,
\begin{align}
\langle i\mathbb{H} \rangle_{\textrm{Lie}} := \textrm{span}_{\mathbb{R}} \left( \langle i\mathbb{H}\rangle_{[\cdot,\cdot]} \right).
\end{align}
A set $i\mathbb{H} \subseteq \mathfrak{su}(2^{n})$ is a {\bf universal generating set} for $\mathfrak{su}(2^{n})$ if and only if its Lie closure is $\mathfrak{su}(2^{n})$, i.e.,
\begin{align}
\langle i\mathbb{H} \rangle_{\textrm{Lie}} =\mathfrak{su}(2^{n}).
\end{align}
Let us now state more formally the observation made in \Cref{subsec:results}:
\begin{Observation} \label{prop:univ_sets_in_Lie_closure_univ_sets} Let $i\mathbb{H'} \subseteq \mathfrak{su}(2^{n})$ be a universal generating set for $\mathfrak{su}(2^{n})$. Then $i\mathbb{H} \subseteq \mathfrak{su}(2^{n})$ is a universal generating set for $\mathfrak{su}(2^{n})$ if and only if 
\begin{align}
i\mathbb{H'} \subseteq \langle i\mathbb{H} \rangle_{\textrm{Lie}}.
\end{align}
\end{Observation}
The proof of this observation is straightforward: For any $i\mathbb{H} \subseteq \mathfrak{su}(2^{n})$, it must be the case that $\langle i\mathbb{H} \rangle_{\textrm{Lie}} \subseteq \mathfrak{su}(2^n)$, and for any $i\mathbb{H'} \subseteq \mathfrak{su}(2^n)$ for which $i\mathbb{H'} \subseteq \langle i\mathbb{H} \rangle_{\textrm{Lie}}$, it must also be the case that $\langle i\mathbb{H'}\rangle_{\textrm{Lie}} \subseteq \langle i\mathbb{H} \rangle_{\textrm{Lie}}$. Thus, if $\langle i\mathbb{H} \rangle_{\textrm{Lie}} = \mathfrak{su}(2^n)$, then $i\mathbb{H'} \subseteq \langle i\mathbb{H} \rangle_{\textrm{Lie}}$ holds directly, and if $\langle i\mathbb{H'}\rangle_{\textrm{Lie}} = \mathfrak{su}(2^n)$, then $\mathfrak{su}(2^n) \subseteq \langle i\mathbb{H} \rangle_{\textrm{Lie}} \subseteq \mathfrak{su}(2^n)$, ensuring that $i\mathbb{H}$ is a universal generating set. Nevertheless, this observation allows a question regarding the universality of a set $i\mathbb{H}$ to be reduced to a question of whether $\langle i\mathbb{H} \rangle_{\textrm{Lie}}$ contains a set known to be universal. The results established in the next section provide one manner for doing so.

As $\mathfrak{su}(2^{n})$ is a real vector space, it is guaranteed to have a basis, and a particularly useful choice for our purposes is comprised of Pauli strings. Let $I$ denote the $2 \times 2$ identity matrix and let $X, Y, Z$ denote the Pauli-$X$, Pauli-$Y$ and Pauli-$Z$ operators respectively. An $n$-qubit {\bf Pauli string} is an operator $iP$ acting on $\mathcal{H}$ given by
\begin{align}
    iP = i\bigotimes_{j=1}^{n}p_{j} 
\end{align}
where, for each $j$, $p_{j} \in \{I,X,Y,Z\}$. We denote the set of all $n$-qubit Pauli strings by 
\begin{align}
i\mathcal{P}_{n} := \left\{i\bigotimes_{j=1}^{n}p_{j} | p_{j} \in \{I,X,Y,Z\}, j = 1, \dots,n \right\}
\end{align}
and define $i\mathcal{P}_{n}^{*} := i\mathcal{P}_{n} \setminus \{iI \otimes I \otimes \dots \otimes I\}$\footnote{A further remark regarding notation. We will always denote elements of $i\mathcal{P}_{n}^{*}$ with capital letters, i.e., as $iP$ or $iA_{k}$. All tensor factors of these elements, which are elements of $\{I,X,Y,Z\}$, will be represented with lower case letters, such as $p_{j}$ or $a_{j}$. This will allow us to distinguish between, say, an element $iA_{k} \in i\mathcal{P}_{n}^{*}$ which is an operator on $n$-qubits, with its tensor factors $a_{j}$, as sometimes the indices $j$ and $k$ will run over the same set.}. Note that, since the Pauli operators $X,Y,Z$ are Hermitian and traceless, $i\mathcal{P}_{n}^{*} \subset \mathfrak{su}(2^{n})$. Also, since either $AB = BA$ or $AB = -BA$ for all $A, B \in\{I,X,Y,Z\}$, it follows that, for any $iP, iQ \in i\mathcal{P}_{n}$, either they commute, i.e., $(iP)(iQ) = (iQ)(iP)$, or they {\bf anti-commute}, i.e., $(iP)(iQ)=-(iQ)(iP)$. In the latter case, we have that
\begin{align}
    \textrm{ad}_{iP}(iQ) = (iP)(iQ) = \pm iR,
\end{align}
where $iR$ is also a Pauli string. It is here where the rationale for including the factor of a half in the definition of the adjoint map becomes clear: it allows any set $i\mathbb{P} \subseteq i\mathcal{P}_{n}$ to essentially map to $i\mathcal{P}_{n}$ under the adjoint map without having to keep track of real coefficients other than $\pm 1$. To make this more precise, let us introduce the following notion of proportionality. Let $iP \in i\mathcal{P}_{n}$ and $iH \in \mathfrak{su}(2^n)$. We say that $iP$ and $iH$ are {\bf proportional}, which we denote by $iH \propto iP$, if either $iH = iP$ or $iH = -iP$. Similarly, for sets $i\mathbb{P} \subseteq i\mathcal{P}_{n}^*$ and $i\mathbb{H} \subseteq \mathfrak{su}(2^n)$, we write $i\mathbb{H} \propto i\mathbb{P}$ if, for each $iH \in i\mathbb{H}$ there exists an element $iP \in i\mathbb{P}$ such that $iH \propto iP$. 

Consider a set $i\mathbb{P} \subseteq i\mathcal{P}_{n}^{*} \subseteq \mathfrak{su}(2^n)$. From earlier, we know that $i\mathbb{P}$ is a universal generating set for $\mathfrak{su}(2^n)$ if and only if $\langle i\mathbb{P} \rangle_{\textrm{Lie}} = \mathfrak{su}(2^n)$. However, we also know that $\mathfrak{su}(2^n) = \textrm{span}_{\mathbb{R}} i\mathcal{P}_{n}^{*}$, that $\langle i\mathbb{P} \rangle_{\textrm{Lie}} = \textrm{span}_{\mathbb{R}} \langle i\mathbb{P} \rangle_{[\cdot,\cdot]}$, and, from the above paragraph, that $\langle i\mathbb{P} \rangle_{[\cdot,\cdot]} \subseteq i\mathcal{P}_{n}^{*}$. Putting these facts together, we arrive at one of the reasons why demonstrating the universality of sets of Pauli strings remains tractable:
\begin{Proposition}[Universality of Pauli strings~\cite{smith2025optimally}]\label{prop:universal_pauli_commutator}
    Let $i\mathbb{P}\subseteq i\mathcal{P}_{n}^{*}$. Then, 
    \begin{align}
        &\langle i\mathbb{P}\rangle_{\mathrm{Lie}}=\mathfrak{su}(2^n)\nonumber 
        \Leftrightarrow \langle i\mathbb{P}\rangle_{[\cdot,\cdot]} \propto i\mathcal{P}_n^*.
    \end{align}
\end{Proposition}
That is, to demonstrate the universality of a set $i\mathbb{P}\subseteq i\mathcal{P}_{n}^{*}$, one needs only to consider how $i\mathbb{P}$ behaves under the adjoint map associated to elements therein. In their essence, many of the results considered in this work boil down to showing the right-hand side of the above proposition for a suitable choice of $i\mathbb{P}$.

Several further concepts related to Pauli strings will be needed below. For an arbitrary $iH \in \mathfrak{su}(2^n)$, let us define the {\bf Pauli set} of $iH$, denoted $i\mathbb{P}(iH)$, to be the elements of $i\mathcal{P}_{n}^{*}$ appearing in the decomposition of $iH$ in the Pauli string basis with a non-zero coefficient. That is, if 
\begin{align}
    iH = \sum_{iP \in i\mathcal{P}_{n}^{*}} \alpha_{iP}iP,
\end{align}
then 
\begin{align}
    i\mathbb{P}(iH) := \{iP \in i\mathcal{P}_{n}^{*} |\alpha_{iP} \neq 0 \}.
\end{align}
Letting $[n] := \{1,\dots,n\}$, we define the {\bf support} of a Pauli string $iP = i\bigotimes_{j=1}^{n}p_{j} \in i\mathcal{P}_{n}^{*}$, denoted \textrm{supp}(iP), via
\begin{align}
\textrm{supp}(iP) := \{j \in [n]| p_{j} \neq I \}.
\end{align}
For a set of Pauli strings $i\mathbb{P} \subseteq i\mathcal{P}_{n}^{*}$, the support of $i\mathbb{P}$ is given by the union of the supports of its constituents:
\begin{align}
\textrm{supp}(i\mathbb{P}) := \bigcup_{iP \in i\mathbb{P}} \textrm{supp}(iP).
\end{align}
Let $K \subseteq [n]$ denote a subset of qubit labels. We define the {\bf restriction} of a Pauli string $iP \in i\mathcal{P}_{n}^{*}$ to $K$, denoted $iP_{|K}$, to be the $|K|$-qubit Pauli string that remains after removing the tensor factors not in $K$, i.e.,
\begin{align}
iP_{|K}:=i\bigotimes_{j \in K} p_{j}.
\end{align}
As with the notion of support above, the notion of restriction to $K$ can be extended to a set of Pauli strings $i\mathbb{P}$:
\begin{align}
i\mathbb{P}_{|K}:=\{iP_{|K} | iP \in i\mathbb{P} \} \setminus \left\{i \bigotimes_{j \in K} I \right\}.
\end{align}

In addition to the set $\langle i \mathbb{P}\rangle_{[\cdot, \cdot]}$, it is useful to define the following set of operators acting on $\mathcal{H}$:
\begin{align}
    \langle i\mathbb{P}\rangle_{\times} := \{\pm 1, \pm i\} \times \bigcup_{r=1}^{\infty}\left\{\prod_{j=1}^{r} iP_{j} | iP_{j} \in i\mathbb{P}\right\}.
\end{align}
We call a set $i\mathbb{P}\subseteq i\mathcal{P}_n^*$ {\bf product universal} for $i\mathcal{P}_n$ if and only if
    \begin{align}
        i\mathcal{P}_n\subseteq\langle i\mathbb{P}\rangle_{\times}.
    \end{align}
Note that, in general, $\langle i\mathbb{P}\rangle_{\times}\nsubseteq\mathfrak{su}(2^n)$, since products of skew-Hermitian Pauli strings need not be skew-Hermitian.

Finally, we also require some concepts from graph theory \cite{diestel2025graphTheory}. A {\bf graph} $G$ is a pair of set $(V, E)$ where $V$ is a set of {\bf vertices} and $E$ is a set of unordered pairs of vertices called {\bf edges}.  The graph $G=(V,E)$ {\bf connected} if, for any distinct vertices $v,w\in V$, there exists a sequence of vertices $u_{1}, \dots, u_{r} \in V$ such that $u_{1} = v$ and $u_{r} = w$, and a sequence of edges $e_{1}, \dots, e_{r-1}\in E$ such that, for all $j = 1, \dots, r-1$, $u_{j}, u_{j+1} \in e_{j}$. For $v\in V$, the {\bf neighborhood} of $v$ in $G$ is the set of vertices
    \begin{align}
        N_{G}(v):=\{w\in V|\{v,w\}\in E\}.
    \end{align}
Let $W \subseteq V$ be a subset of vertices. A vertex $v \in V \setminus W$ is said to have a {\bf unique neighborhood in $W$} if $N_{G}(v) \cap W \neq \emptyset$ and for all other $u \in V \setminus W$ such that $u \neq v$, $N_{G}(u) \cap W \neq N_{G}(v) \cap W$. Let $N_{G}^{\textrm{uniq}}(W) \subseteq V \setminus W$ denote the set of all vertices with a unique neighborhood in $W$. The {\bf expansion of $W$ by its unique neighbors in $V$} is the set $\mathcal{V}_{G}(W)$ defined via 
    \begin{align}
        \mathcal{V}_{G}(W)&:= W \cup N_{G}^{\textrm{uniq}}(W).
    \end{align}
For what follows, it is also convenient to introduce notation for the {\bf $r$-fold expansion of $W$ by its unique neighbors in $V$}, i.e.,
    \begin{align}
        \mathcal{V}^{(r)}_G(W) := \underbrace{\mathcal{V}_G \circ \cdots \circ \mathcal{V}_G}_{r \textrm{ times}}(W).
    \end{align}

Given a set of Pauli strings $i\mathbb{P} \subseteq i\mathcal{P}_{n}^{*}$, we can define the {\bf anti-commutation graph} of $i\mathbb{P}$, denoted $G(i\mathbb{P}):=(V(i\mathbb{P}),E(i\mathbb{P}))$, where the set of vertices $V(i\mathbb{P})$ are identified with the set of Pauli strings $i\mathbb{P}$ and where the set of edges is given by
\begin{align}
    E(i\mathbb{P}):=\{\{iP_1,iP_2\}|iP_1,iP_2\in V(i\mathbb{P}), [iP_1,iP_2]\neq 0\}.
\end{align}

\section{Results, Formally Stated} \label{sec:formal_results}

In this section, we formally state the results presented in \Cref{subsec:results}. The proofs of these statements are to be found in the corresponding appendices, as indicated below each statement.

To begin with, let us consider \Cref{res:Pauli_strings_only} pertaining to the universality of sets of Pauli strings, which makes use of the notions of product universality and anti-commutation graphs presented in the preceding section. 
\begin{Theorem}[Pauli string universality]\label{thm:nec_suf_pauli}
Let $i\mathbb{P}\subseteq i\mathcal{P}_n^*$ with $n \ge 3$. Then $\langle i\mathbb{P}\rangle_{\textrm{Lie}}=\mathfrak{su}(2^n)$ if and only if all of the following hold:
    \begin{enumerate}
        \item \label{item:thm_1} there exists a set $i\mathbb{Q}\subseteq \langle i\mathbb{P}\rangle_{[\cdot,\cdot]}$ such that $\langle i\mathbb{Q}_{|\mathrm{supp}(i\mathbb{Q})}\rangle_{[\cdot,\cdot]}\propto i\mathcal{P}_k^*$ for $2\leq k< n$ where $k:=|\mathrm{supp}(i\mathbb{Q})|$.
        \item \label{item:thm_2}$i\mathcal{P}_{n-k} \subseteq \langle i\mathbb{P}_{|[n]\setminus\mathrm{supp}(i\mathbb{Q})}\rangle_\times$.
        \item \label{item:thm_3}$G(i\mathbb{P})$ is connected.
    \end{enumerate}
\end{Theorem}
The proof of this theorem is provided in \cref{app:proof_thm:nec_suf_pauli}.
Importantly, the proof of sufficiency is constructively tied to Theorem~1 in Ref.~\cite{smith2025optimally}. Therefore, one can also use the algorithm there to efficiently generate any Pauli string from a universal set (up to a possible overhead given by the construction of the set $i\mathbb{S}$ in Eq.~\ref{eq:spanning-tree-set_app}). Furthermore, recent work by Cuypers (see Ref.~\cite{cuypers2026dynamicallie}) has demonstrated necessary and sufficient conditions for the universality of sets of Pauli strings $i\mathbb{P}$ based on the inclusion of any of a set of $32$ distinguished subgraphs in $G(i\mathbb{P})$ (the reader is referred to Ref.~\cite{cuypers2026dynamicallie} for the details). In particular, this leads to an algorithmic approach to checking universality in this context. One possible avenue for future work could be to see if these approaches can be fruitfully combined to, e.g., search for optimally efficient generating sets of Pauli strings.

In the remainder of this section, we consider several other generating sets with more general Hamiltonians, whose universality can be shown by leveraging \Cref{thm:nec_suf_pauli}. To aid us in doing so, we have the following:
\begin{Lemma}[Pauli isolation]\label{lem:isolation}
    Let $i\mathbb{Q}\subseteq i\mathcal{P}_n^*$ and let $iH \in \mathfrak{su}(2^n)$ be given by
\begin{align*}
iH=\sum_{iP\in i\mathcal{P}_n^*}\alpha_{iP}\, iP,
\end{align*}
with all coefficients $\alpha_{iP}$ being real. Let 
\begin{align}
    i\mathbb{P}(iH):=\{iP\in i\mathcal{P}_n^*\mid \alpha_{iP}\neq 0\}
\end{align}
and consider the anti-commutation graph 
\begin{align}
    G:=G\!\big(i\mathbb{Q}\cup i\mathbb{P}(iH)\big).
\end{align}
For $iP\in i\mathcal{P}_n^*$, if $iP\in \mathcal{V}_{G}(i\mathbb{Q})$, then $iP\in \langle i\mathbb{Q}\cup\{iH\}\rangle_{\mathrm{Lie}}$.
\end{Lemma}
The proof of this lemma is given in \cref{app:proof_lem:isolation}. In its essence, this lemma allows, in certain circumstances, questions of the universality of set $i\mathbb{H}$ involving both Hamiltonian and Pauli string terms to be reduced to questions of universality of sets of Pauli strings only\footnote{Note that this isolation lemma can also be interpreted as a character-weighted Pauli twirl~\cite{viola2003robust}, generalized to the commutant.}. By applying this lemma iteratively, we can arrive at the following result:

\begin{Theorem}[Hamiltonian universality]\label{thm:suf_ham}
Let $i\mathbb{Q}\subseteq i\mathcal{P}_n^*$ and let $iH \in \mathfrak{su}(2^n)$ be given by
\begin{align*}
iH=\sum_{iP\in i\mathcal{P}_n^*}\alpha_{iP}\, iP,
\end{align*}
with all coefficients $\alpha_{iP}$ being real. Let 
\begin{align}
    i\mathbb{P}(iH):=\{iP\in i\mathcal{P}_n^*\mid \alpha_{iP}\neq 0\}
\end{align}
and consider the anti-commutation graph 
\begin{align}
    G:=G\!\big(i\mathbb{Q}\cup i\mathbb{P}(iH)\big).
\end{align}

If, for some $r \in \mathbb{N}$, $\langle\mathcal{V}_G^{(r)}(i\mathbb{Q})\rangle_{\textrm{Lie}}=\mathfrak{su}(2^n)$ then $\langle i\mathbb{Q}\cup \{iH\}\rangle_{\textrm{Lie}}=\mathfrak{su}(2^n)$. 

If $\mathcal{V}_G^{(r)}(i\mathbb{Q}) = i\mathbb{Q}\cup i\mathbb{P}(iH)$, then the reverse implication also holds: if $\langle i\mathbb{Q}\cup \{iH\}\rangle_{\textrm{Lie}}=\mathfrak{su}(2^n)$ then $\langle\mathcal{V}_G^{(r)}(i\mathbb{Q})\rangle_{\textrm{Lie}}=\mathfrak{su}(2^n)$.
\end{Theorem}
The proof of this theorem is given in \cref{app:proof_thm:suf_ham}.

Using the above results, it is possible to derive further necessary and sufficient conditions in certain situations where the set of controls $i\mathbb{H}$ contains local control on some or all of the qubits. These results have been considered previously in the literature; see Ref.~\cite{bremner2004fungible} and Ref.~\cite{zeier2011symmetry} respectively. Here, we show that these results are consequences of the common framework encapsulated by the results presented above. 

Let us first consider the case where local control is present on all $n$ qubits. Let us define $i\mathbb{P}_{\text{s.q.}}$ to be the set of all unit support Pauli strings containing either a Pauli-$X$ term or a Pauli-$Z$ term, i.e., 
\begin{align}
i\mathbb{P}_{\text{s.q.}} := \{iX_{\{j\}}, iZ_{\{j\}} : j \in [n]\},
\end{align}
where $iX_{\{j\}}$ denotes the Pauli string $iP \in i\mathcal{P}_{n}^*$ such that $p_{j} = X$ and $p_{l} = I$ for all $l \neq j$, and similarly for $iZ_{\{j\}}$. Let us also denote by $\bigotimes_{j \in [n]}\mathfrak{su}(2)$ the Lie algebra generated by $i\mathbb{P}_{\text{s.q.}}$, which corresponds to being able to locally control each qubit arbitrarily. We have the following result, which was previously established in Ref.~\cite{bremner2004fungible}, which provides the necessary and sufficient conditions for an additional Hamiltonian to promote all local control to universal generating set:
\begin{Corollary} \label{cor:local_control}  Let $i\mathbb{H}_{\text{s.q.}}\subseteq \mathfrak{su}(2^{n})$ be such that $\langle i\mathbb{H}_{\text{s.q.}} \rangle_{\textrm{Lie}} = \bigotimes_{j \in [n]}\mathfrak{su}(2)$. Let $iH\in\mathfrak{su}(2^n)$ and define $i\mathbb{H} := i\mathbb{H}_{\text{s.q.}} \cup \{iH\}$. The set $i\mathbb{H}$ is a universal generating set for $\mathfrak{su}(2^{n})$ if and only if the following hold:
\begin{enumerate}
    \item \label{item:cor_lc_1} there exists an element $iP \in i\mathbb{P}(iH)$ such that $|\textrm{supp}(iP)|$ is even,
    \item \label{item:cor_lc_3}  the anti-commutation graph $G(i\mathbb{P}(iH)\cup i\mathbb{P}_{\text{s.q.}})$ is connected.
\end{enumerate} 
\end{Corollary}
The proof of this corollary is given in \Cref{app:cor_local_control} and essentially leverages \Cref{thm:suf_ham} to reduce the problem to a case where \Cref{thm:nec_suf_pauli} can be applied.

For the case where only a subset of qubits can be controlled locally, let us define
\begin{align}
i\mathbb{P}_{\text{s.q.}}^{(k)} := \{iX_{\{j\}}, iZ_{\{j\}} : j \in [k]\},
\end{align}
using the same notation as above (in particular, the strings $iX_{\{j\}}$ and $iZ_{\{j\}}$ are still Pauli strings on $n$ qubits). Let us also define $\bigotimes_{j \in [k]}\mathfrak{su}(2) := \langle i\mathbb{P}_{\text{s.q.}}^{(k)}\rangle_{\textrm{Lie}}$. Using the results above, we obtain the following corollary, which is a result previously considered in Ref.~\cite{zeier2011symmetry}. Below, we are denoting by $iX_{j}X_{j+1}$ the Pauli string containing Pauli-$X$ terms and locations $j$ and $j+1$ with identities elsewhere, and similarly for $iY_{j}Y_{j+1}$ and $iZ_{j}Z_{j+1}$.
\begin{Corollary}\label{cor:heisenberg}
    Let
    \begin{align*}
        iH=i\sum_{j=1}^{n-1} J_xX_jX_{j+1}+J_yY_jY_{j+1}+J_zZ_jZ_{j+1}
    \end{align*}
    for $J_x,J_y,J_z\in\mathbb{R}^*$ and let $i\mathbb{H}^{(k)}_{\text{s.q.}}\subseteq \mathfrak{su}(2^{n})$ be such that $\langle i\mathbb{H}^{(k)}_{\text{s.q.}} \rangle_{\textrm{Lie}} =\bigotimes_{j \in [k]}\mathfrak{su}(2)$ for some $k \in [n]$. Then, $\langle i\mathbb{H}^{(k)}_{\text{s.q.}} \cup \{iH\}\rangle_{\mathrm{Lie}}=\mathfrak{su}(2^n)$ if and only if $k \ge 2$.
\end{Corollary}
The proof of this corollary is given in \Cref{app:cor_heisenberg} and again proceeds by reducing the situation to a case where \Cref{thm:nec_suf_pauli} can be applied.

\section{Acknowledgments}
This project was funded in whole or in part by the Austrian Science Fund (FWF) [SFB BeyondC F7102, DOI: 10.55776/F71; WIT9503323, DOI: 10.55776/WIT9503323]. For open access purposes, the authors have applied a CC BY public copyright license to any author-accepted manuscript version arising from this submission.
This work was also supported by the Austrian Research Promotion Agency (FFG) and the European Union via NextGeneration EU under Contract Number FO999921407 (HDcode)
This work was funded by the European Union (ERC, QuantAI, Project No. 101055129). 
Views and opinions expressed are however those of the author(s) only and do not necessarily reflect those of the European Union or the European Research Council. Neither the European Union nor the granting authority can be held responsible for them. I.D.S acknowledges financial support through the Government of Spain (Severo Ochoa CEX2019-000910-S and FUNQIP), Fundació Cellex, Fundació Mir-Puig, Generalitat de Catalunya (CERCA program), and the AXA Chair in Quantum Information Science.

\bibliography{bib}

\appendix
\onecolumngrid
\newpage

\section{Proof of Theorem~\ref{thm:nec_suf_pauli}}\label{app:proof_thm:nec_suf_pauli}

In this section, we proof \Cref{thm:nec_suf_pauli}, which we restate below for convenience. Prior to doing so, we require further concepts and results relating to Lie algebra theory and graph theory from the literature, which we briefly introduce. Thereafter, we establish several lemmas required for the proof of the theorem, before turning to the proof of the theorem itself.

From graph theory \cite{diestel2025graphTheory}, we require the following concepts in addition to those presented in \Cref{sec:background}. Let $G=(V,E)$ be a graph. Two vertices $v_1,v_2$ are {\bf adjacent} if there exist $e\in E$ such that $e=\{v_1,v_2\}$. A {\bf path} between two vertices $v,w\in V$ is a sequence of vertices $\gamma_G(v,w):=(u_1,...,u_r)\in V^{r}$ such that $u_1=v$, $u_r=w$ and $u_i$ is adjacent to $u_{i+1}$ for $i=1,...,r-1$. We define $l(\gamma_G(v,w)):=r-1$ as the {\bf length} of a path $\gamma_G(v,w)$ (i.e., the number of edges it contains). A {\bf shortest path} $\gamma^*_G(v,w)$ between two vertices $v,w\in V$ is a path between $v$ and $w$ with the smallest length, i.e., $\gamma^*_G(v,w) = \textrm{argmin}_{\gamma_G(v,w)} l(\gamma_G(v,w))$. Note that shortest paths are not necessarily unique. The {\bf graph distance} between two vertices, denoted $d_G(v,w)$, is the length of the shortest paths, i.e. $d_G(v,w) = l(\gamma^*_G(v,w))$.

Suppose that $G$ is also connected. The graph $G$ is a {\bf tree} if, for every pair of vertices $v,w\in V$, the path $\gamma_G(v,w)$ is unique. A {\bf spanning tree} of $G$, denoted $T_G$, is a subgraph of $G$ that is a tree and contains all vertices of $G$. For a fixed $v \ in V$, a spanning tree $T_G$ is a {\bf shortest-path spanning tree} of $G$ rooted at $q$ if, for all $w \in V$,
        \begin{align}
            l(\gamma_{T_G}(v,w))=d_G(v,w).
        \end{align}
Note that shortest-path trees are also not unique. They are, however, guaranteed to exist:
\begin{Proposition}[Existence of shortest-path spanning trees (\cite{erickson2019algorithms}, Chapter~8)] \label{prop:exist_spanning_tree}
    Let $G=(V,E)$ be a connected graph and $v\in V$ fixed. There exists a shortest-path spanning tree $T_G$ of $G$ rooted at $v$.
\end{Proposition}

In the proof of \Cref{thm:nec_suf_pauli}, and through much of the rest of this manuscript, we consider the anti-commutation graph $G(i\mathbb{P})$ of a set of Pauli strings $i\mathbb{P} \subseteq i\mathcal{P}_{n}^*$. The following lemma relating shortest paths in $G(i\mathbb{P})$ to composed adjoint maps is useful later:
\begin{Lemma}\label{lem:path_adjoint_app}
        Let $i\mathbb{P}\subseteq i\mathcal{P}_n^*$ be such that $G:=G(i\mathbb{P})$ is connected. 
        Let $\gamma^*_G(iQ,iP):=(iP_1,...,iP_r)\in V^{r}$ with $iP_1=iQ$ and $iP_r=iP$ be a shortest path in $G$ between vertices $iQ\neq iP$.

        Then,
        \begin{align}
        ad_{iQ}\circ ad_{iP_2}\circ\cdots\circ ad_{iP_{r-1}}(iP)\neq 0.
        \end{align}
    \end{Lemma}

    \begin{proof}
    Observe that, for any $iP,iQ\in i\mathcal{P}_n^*$, the adjoint map $\textrm{ad}_{iP}(iQ)$ is either $0$, if $iP$ and $iQ$ commute, or equal to their product $(iP)(iQ)$, if they anti-commute. Accordingly,
    \begin{align}\label{eq:commutator_product}
         ad_{iQ}\circ ad_{iP_2}\circ\cdots\circ ad_{iP_{r-1}}(iP)=
         \begin{cases}
            \prod_{k=1}^r iP_k\\
             0 
         \end{cases}
    \end{align}
    where the product is taken to be ordered.

    Let $iP, iQ \in i\mathbb{P}$ and let $\gamma^*_G(iQ,iP):=(iP_1,...,iP_r)\in V^{r}$ be a shortest path in $G$ between $iQ =: iP_{1}$ and $iP =: iP_{r}$. Since $(iP_1,...,iP_r)$ is a path, it must be that $\{iP_{j}, iP_{j+1}\}$ is an edge in $G$ for all $j \in [r-1]$, which in turn means that $\textrm{ad}_{iP_{j}}(iP_{j+1}) \neq 0$.
    
    Suppose, for a contradiction, that $ad_{iQ}\circ ad_{iP_2}\circ\cdots\circ ad_{iP_{r-1}}(iP)= 0$. Then, there exists some $j \in [r-1]$ such that $[iP_j, \prod_{j<l\leq r}iP_l]=0$. If $j=r-1$, this means that $[iP_{r-1}, iP_{r}] = 0$, contradicting the assumption that $\{iP_{r-1}, iP_{r}\}$ is an edge in $G$. If $j<r-1$, it must be that $iP_{j}$ anti-commutes with $k = 0 \bmod 2$ elements of $\{ iP_{j+1}, \dots, iP_{r}\}$. As stated above, $\textrm{ad}_{iP_{j}}(iP_{j+1}) \neq 0$, it must be that $k > 1$ and hence there must exist some $l \in \{j+1, \dots, r\}$ be such that $iP_{j}$ and $iP_{l}$ anti-commute. It follows that $\{iP_{j}, iP_{l}\}$ is an edge in $G$, and thus $(iP_{1},\dots, iP_{j},iP_{l} iP_{l+1}, \dots, iP_{r})$ is a path in $G$ between $iQ$ and $iP$. However, $(iP_{1},\dots, iP_{j},iP_{l} iP_{l+1}, \dots, iP_{r})$ is strictly shorter than $\gamma^*_G(iQ,iP)$, contradicting the assumption that the latter is a shortest path.
    \end{proof}

Let us also establish the following lemma:
\begin{Lemma} \label{lem:last_anti_commute} Let $r\ge 2$ and let $iA, iB_{1},\dots,iB_{r} \in i\mathcal{P}_{n}^*$ be such that 
\begin{align}
   0 \neq iA \propto \textrm{ad}_{iB_{1}} \circ \dots \circ \textrm{ad}_{iB_{r-1}}(iB_{r}).
\end{align}
Then $\textrm{ad}_{iB_{1}}(iA) \neq 0$.
\end{Lemma}

\begin{proof} Let $r$ and $iA, iB_{1},\dots,iB_{r}$ be as in the statement of the lemma. Suppose that $\textrm{ad}_{iB_{1}}(iA) = 0$. Since $iA \propto \textrm{ad}_{iB_{1}} \circ \dots \circ \textrm{ad}_{iB_{r-1}}(iB_{r}) = (iB_{1})\prod_{j=2}^{r}(iB_{j})$ and $iB_{1}$ commutes with itself, for $\textrm{ad}_{iB_{1}}(iA) = 0$ it must be the case that 
\begin{align}
    [iB_{1}, \prod_{j=2}^{r}(iB_{j})] = 0.
\end{align}
However, this would entail that $\textrm{ad}_{iB_{1}} \circ \dots \circ \textrm{ad}_{iB_{r-1}}(iB_{r}) = 0$ and hence also $iA = 0$, forming a contradiction.
\end{proof}

From Lie algebra theory, we require the following concepts in addition to those presented in \Cref{sec:background}. Let $\mathfrak{b}$ denote a Lie algebra. An {\bf ideal} of $\mathfrak{b}$, denoted $\mathfrak{a}$, is a linear subspace of $\mathfrak{b}$ such that, for all $a \in \mathfrak{a}$ and all $b \in \mathfrak{b}$, $[a,b] \in \mathfrak{a}$. The ideal $\mathfrak{a}$ is a {\bf proper ideal} of $\mathfrak{b}$ if $\mathfrak{a}$ is a strict subspace of $\mathfrak{b}$. A Lie algebra $\mathfrak{b}$ is {\bf non-Abelian} if there exists some $a,b \in \mathfrak{b}$ such that $[a,b] \neq 0$ and $\mathfrak{b}$ is {\bf simple} if it is non-Abelian and contains no non-zero proper ideals. The following is a known result from the literature: 
\begin{Proposition}[$\mathfrak{su}(N)$ is simple~\cite{knapp2002liegroups}]\label{prop:no_ideals}
    The Lie algebra $\mathfrak{su}(N)$ is simple for $N\geq 2$.
\end{Proposition}
It follows from this that our Lie algebras of interest, namely $\mathfrak{su}(2^n)$ for $n \ge 2$, is simple.

Finally, we require one further existing result from the literature on universality of Pauli strings, which we state here using the notation established in this manuscript:
\begin{Proposition}[Theorem~1 in Ref.~\cite{smith2025optimally}]\label{thm:smith2025optimally}
    Let $i\mathbb{Q} \subseteq i\mathcal{P}^*_{n}$ be such that $\langle i\mathbb{Q}_{|\mathrm{supp}(i\mathbb{Q})}\rangle_{[\cdot,\cdot]}\propto i\mathcal{P}_k^*$ for $2\leq k\leq n$.
    Let $i\mathbb{S} \subseteq i\mathcal{P}_{n}^{*}$ be such that $\langle i\mathbb{S}_{|[n]\setminus\mathrm{supp}(i\mathbb{Q})}\rangle_\times \supseteq \mathcal{P}_{n-k}$
    and for all $iR\in i\mathbb{S}$ there exists a $j\in \mathrm{supp}(i\mathbb{Q})$ such that $j\in\mathrm{supp}(i R)$. Then $\langle i\mathbb{Q} \cup i\mathbb{S}\rangle_{\mathrm{Lie}}=\mathfrak{su}(2^n)$, i.e., it is a universal generating set for $\mathfrak{su}(2^{n})$.
\end{Proposition}

Let us now turn to \Cref{thm:nec_suf_pauli}, which we restate here:
\begin{Theorem}[Pauli string universality]\label{thm:nec_suf_pauli_app}
    Let $i\mathbb{P}\subseteq i\mathcal{P}_n^*$ with $n \ge 3$. Then $\langle i\mathbb{P}\rangle_{\textrm{Lie}}=\mathfrak{su}(2^n)$ if and only if all of the following hold:
    \begin{enumerate}
        \item \label{item:thm_1} there exists a set $i\mathbb{Q}\subseteq \langle i\mathbb{P}\rangle_{[\cdot,\cdot]}$ such that $\langle i\mathbb{Q}_{|\mathrm{supp}(i\mathbb{Q})}\rangle_{[\cdot,\cdot]}\propto i\mathcal{P}_k^*$ for $2\leq k< n$ where $k:=|\mathrm{supp}(i\mathbb{Q})|$.
        \item \label{item:thm_2}$i\mathcal{P}_{n-k} \subseteq \langle i\mathbb{P}_{|[n]\setminus\mathrm{supp}(i\mathbb{Q})}\rangle_\times $.
        \item \label{item:thm_3}$G(i\mathbb{P})$ is connected.
    \end{enumerate}
\end{Theorem}

\begin{proof}
\textbf{Sufficiency:} 

Suppose that \Cref{item:thm_1,item:thm_2,item:thm_3} from the statement of the theorem hold. Note that $i\mathbb{Q}$ satisfies the conditions of \Cref{thm:smith2025optimally}, so to prove sufficiency, we can proceed by demonstrating that there exists a set $i\mathbb{S} \subseteq \langle i\mathbb{P}\rangle_{[\cdot,\cdot]}$ which satisfies the remaining conditions of \Cref{thm:smith2025optimally}.

Since we require that $\langle i\mathbb{Q}_{|\mathrm{supp}(i\mathbb{Q})}\rangle_{[\cdot,\cdot]}\propto i\mathcal{P}_k^*$, i.e., equality of these sets up to $\pm1$ for each element, let us ignore  any potential $-1$ coefficients on the Pauli strings in $i\mathbb{Q}$ for simplicity of notation. Since $i\mathbb{Q}\subseteq \langle i\mathbb{P}\rangle_{[\cdot,\cdot]}$, then for each $iQ \in i\mathbb{Q}$ either (i) $iQ \in i\mathbb{P}$ directly or (ii) there exists an $r \ge 2$ and elements $iP_{1},\dots,iP_{r} \in i\mathbb{P}$ such that 
\begin{align}
    iQ = \textrm{ad}_{iP_{1}}\circ \dots \circ \textrm{ad}_{iP_{r-1}}(iP_{r}).
\end{align}
Let $\tilde{G} := G(i\mathbb{P} \cup i\mathbb{Q})$, i.e., it is the anti-commutation graph of $i\mathbb{P} \cup i\mathbb{Q}$. Note that $G(i\mathbb{P})$ is a subgraph of $\tilde{G}$ and recall that, by assumption, $G(i\mathbb{P})$ is connected. By \Cref{lem:last_anti_commute}, we know that for all non-zero $iQ$ falling into case (ii) above, we have that $[iQ, iP_{1}] \neq 0$, meaning that the edge $\{iQ, iP_{1} \}$ exists in $\tilde{G}$. Thus every $iQ \in i\mathbb{Q}\setminus i\mathbb{P}$ is connected to the subgraph $G(i\mathbb{P})$ by at least one edge, and so $\tilde{G}$ itself is connected.

Let $iQ \in i\mathbb{Q}$ be fixed and let $T_{\tilde{G}}$ denote a shortest-path spanning tree of $\tilde{G}$ rooted at the vertex corresponding to $iQ$, which we know must exist by \Cref{prop:exist_spanning_tree}. For each $iP \in i\mathbb{P}\cup i\mathbb{Q}$ such that $iP \neq iQ$, let $\gamma^*_{\tilde{G}}(iQ,iP):=(iP_1,...,iP_r)\in V^{r}$ with $iP_1=iQ$ and $iP_r=iP$ denote the path between $iQ$ and $iP$ in $T_{\tilde{G}}$, meaning also that it defines a shortest path between vertices $iQ\neq iP$ in $\tilde{G}$. By \Cref{lem:path_adjoint_app}, we know that, for each $iP \in i\mathbb{P}\cup i\mathbb{Q}$ such that $iP \neq iQ$,
\begin{align}
    \textrm{ad}_{iQ} \circ \textrm{ad}_{iP_2}\circ \dots \circ \textrm{ad}_{iP_{r-1}}(iP) \neq 0.
\end{align}
For each $iP \in i\mathbb{P}\cup i\mathbb{Q}$ such that $iP \neq iQ$, let us define $iR_{(iQ,iP)}$ to be the element of $i\mathcal{P}_{n}^*$ such that 
\begin{align}
    iR_{(iQ,iP)} \propto \textrm{ad}_{iQ} \circ \textrm{ad}_{iP_2}\circ \dots \circ \textrm{ad}_{iP_{r-1}}(iP),
\end{align}
and let us define
\begin{align}\label{eq:spanning-tree-set_app}
    i\mathbb{S} := \{iQ\} \cup \left\{iR_{(iQ,iP)} | iP \in i\mathbb{P}\cup i\mathbb{Q}, iP \neq iQ  \right\}.
\end{align}

For $i\mathbb{S}$ to satisfy the conditions of \Cref{thm:smith2025optimally}, we need to show (a) that $i\mathbb{S}$ is product universal on $[n]\setminus \textrm{supp}(i\mathbb{Q})$, i.e., that $\langle i\mathbb{S}_{|[n]\setminus\mathrm{supp}(i\mathbb{Q})}\rangle_\times \supseteq i\mathcal{P}_{n-k}$, and (b) that, for each $iR \in i\mathbb{S}$ there exists a $j \in \textrm{supp}(i\mathbb{Q})$ such that $j \in \textrm{supp}(iR)$. 

By assumption, we have that $\langle i\mathbb{P}_{|[n]\setminus\mathrm{supp}(i\mathbb{Q})}\rangle_\times \supseteq \mathcal{P}_{n-k}$, so to demonstrate (a) it suffices to show that $\langle i\mathbb{P}_{|[n]\setminus\mathrm{supp}(i\mathbb{Q})}\rangle_\times \subseteq \langle i\mathbb{S}_{|[n]\setminus\mathrm{supp}(i\mathbb{Q})}\rangle_\times$. Let $iP \in i\mathbb{P}\setminus \{iQ\}$ and let $\gamma^*_{\tilde{G}}(iQ,iP) = (iP_{1}, \dots, iP_{r})$ be the path between $iQ$ and $iP$ in $T_{\tilde{G}}$. If $r=2$, the path is simply $\gamma^*_{\tilde{G}}(iQ,iP) = (iQ, iP)$, meaning that $iR_{(iQ,iP)} \propto (iQ)(iP)$. Since $(iQ)^2 = - I \otimes \dots \otimes I$ for any Pauli string $iQ$, this means that $(iQ)(iR_{(iQ,iP)}) \propto iP$ and thus $iP \in \langle i\mathbb{S}\rangle_\times$. If $r > 2$, then $\gamma^*_{\tilde{G}}(iQ,iP) = (iQ, ,iP_{2},\dots, iP_{r-1}, iP)$ and 
\begin{align}
    iR_{(iQ,iP)} \propto (iQ)\left(\prod_{j=2}^{r-1} iP_{j}\right)(iP).
\end{align}
Since $T_{\tilde{G}}$ is a tree, it must also be the case that $\gamma^*_{\tilde{G}}(iQ,iP_{r-1}) = (iQ, ,iP_{2},\dots, iP_{r-1})$ and that 
\begin{align}
    iR_{(iQ,iP_{r-1})} \propto \begin{cases} (iQ)(iP_{2}) & \text{ if } r=3, \\
    (iQ)\left(\prod_{j=2}^{r-2} iP_{j}\right)(iP_{r-1}) &\text{ otherwise.} \end{cases}
\end{align}
In either case, since any pair of Pauli strings either commute or anti-commute, we have that
\begin{align}
    \left(iR_{(iQ,iP_{r-1})}\right)\left(iR_{(iQ,iP)}\right) = \alpha iP
\end{align}
for some $\alpha \in \{\pm 1, \pm i\}$. Thus, for all $iP \in i\mathbb{P}$, we have that $iP \in \langle i\mathbb{S}\rangle_\times$ (recall that $iQ$ is trivially in $\langle i\mathbb{S}\rangle_\times$). Furthermore, since for any $A,B \in  \langle i\mathbb{S}\rangle_\times$, we have that $AB \in  \langle i\mathbb{S}\rangle_\times$, it follows that $ \langle i\mathbb{P}\rangle_\times \subseteq  \langle i\mathbb{S}\rangle_\times$. Finally, since, for any $iA, iB \in i\mathcal{P}_{n}^*$ and any $J \subseteq [n]$, we have that $(iA_{|J})(iB_{|J}) = \alpha i (AB)_{|J}$ for some $\alpha \in \{\pm 1, \pm i\}$, it follows that  $\langle i\mathbb{P}_{|[n]\setminus\mathrm{supp}(i\mathbb{Q})}\rangle_\times \subseteq \langle i\mathbb{S}_{|[n]\setminus\mathrm{supp}(i\mathbb{Q})}\rangle_\times$, as required.

To prove (b), let us consider an arbitrary element $iR \in i\mathbb{S}$. By definition, we have that $iR$ is non-zero and either $iR = iQ$ or $iR \propto \textrm{ad}_{iQ} \circ \textrm{ad}_{iP_{2}} \circ \dots \textrm{ad}_{iP_{r-1}}(iP_{r})$ for some $iP_{2}, \dots, iP_{r} \in i\mathbb{P}$. In the first case, $\textrm{supp}(iR) = \textrm{supp}(iQ)$ and hence the fact that $\textrm{supp}(iR) \cap \textrm{supp}(i\mathbb{Q}) \neq \emptyset$ follows directly. For the second case, \Cref{lem:last_anti_commute} ensures that $[iR, iQ] \neq 0$, and thus must have overlapping support (otherwise they would commute).
    
We have thus shown that there exist sets $i\mathbb{Q}, i\mathbb{S} \subseteq \langle i\mathbb{P}\rangle_{\textrm{Lie}}$ for which $\langle i\mathbb{Q} \cup i\mathbb{S}\rangle_{[\cdot, \cdot]} = i\mathcal{P}_{n}^*$, completing the sufficiency direction of the proof. \\

    \textbf{Necessity:} 

    Suppose that $\langle i\mathbb{P}\rangle_{\textrm{Lie}} = \mathfrak{su}(2^n)$. By \Cref{prop:universal_pauli_commutator}, it must be the case that $\langle i\mathbb{P}\rangle_{[\cdot, \cdot]} \propto i\mathcal{P}_{n}^*$. Let us define $i\mathbb{Q}$ via
    \begin{align}
        i\mathbb{Q} := \{ iQ \in \langle i\mathbb{P}\rangle_{[\cdot, \cdot]} | \textrm{supp}(iQ) \subseteq \{1,2\}\}.
    \end{align}
    From the fact that $\langle i\mathbb{P}\rangle_{[\cdot, \cdot]} \propto i\mathcal{P}_{n}^*$, we have that 
    \begin{align}
        i\mathbb{Q} \propto \{ iP \otimes I_{3} \otimes \dots \otimes I_{n}| iP \in i\mathcal{P}_{2}^*\},
    \end{align} 
    from which we get that $\langle i\mathbb{Q}_{|\mathrm{supp}(i\mathbb{Q})}\rangle_{[\cdot,\cdot]}\propto i\mathcal{P}_2^*$ directly. This establishes \Cref{item:thm_1}.

    Let $iR \in i\mathcal{P}_{n-2}^*$. Since $I_{1} \otimes I_{2} \otimes iR \in i\mathcal{P}_{n}^*$ and since  $\langle i\mathbb{P}\rangle_{[\cdot, \cdot]} \propto i\mathcal{P}_{n}^*$, there must exist $iP_{1},\dots,iP_{r} \in i\mathbb{P}$ such that
    \begin{align}
        I_{1} \otimes I_{2} \otimes iR &\propto \textrm{ad}_{iP_1}\circ \cdots \circ \textrm{ad}_{iP_{r-1}}(iP_r) \\
        &= \prod_{j=1}^{r}iP_{j}.
    \end{align}
    It follows that 
    \begin{align}
        iR &= \left(\prod_{j=1}^{r}iP_{j}\right)_{|\{3,\dots,n\}} \\ 
        &\equiv \prod_{j=1}^r i(P_{j})_{|\{3,\dots,n\}}. 
    \end{align}
    Thus, $iR \in \langle i\mathbb{P}_{|\{3,\dots,n\}}\rangle_{\times}$. Since this holds for all $iR \in i\mathcal{P}_{n-2}^*$, we have that $i\mathcal{P}_{n-2}^* \subseteq \langle i\mathbb{P}_{|\{3,\dots,n\}}\rangle_{\times}$, establishing \Cref{item:thm_2}.
    
For the remaining condition, \Cref{item:thm_3}, we consider the contrapositive direction: if $G(i\mathbb{P})$ is {\em not} connected, then $\langle i\mathbb{P}\rangle_{\textrm{Lie}} \neq \mathfrak{su}(2^n)$. Let $i\mathbb{P}_{1}, i\mathbb{P}_{2} \subseteq i\mathcal{P}_{n}^{*}$ be non-empty sets such that $i\mathbb{P} = i\mathbb{P}_{1} \cup i\mathbb{P}_{2}$, $i\mathbb{P}_{1} \cap i\mathbb{P}_{2} = \emptyset$ and for all $iP \in i\mathbb{P}_{1}$ and all $iQ \in i\mathbb{P}_{2}$, $\{iP, iQ\}\notin E(i\mathbb{P})$. That is, if $V_{1}, V_{2} \subset V(i\mathbb{P})$ are the sets of vertices corresponding to $i\mathbb{P}_{1}$ and $i\mathbb{P}_{2}$ respectively, the subgraphs of $G(i\mathbb{P})$ corresponding to $V_{1}$ and $V_{2}$ are disconnected in $G(i\mathbb{P})$. Let $\mathfrak{a} := \langle i\mathbb{P}_{1}\rangle_{\textrm{Lie}}$ denote the Lie algebra associated to $i\mathbb{P}_{1}$. 

We proceed by showing that $\mathfrak{a}$ is a non-zero proper ideal of $\langle i\mathbb{P}\rangle_{\textrm{Lie}}$. By \Cref{prop:no_ideals}, we know that $\mathfrak{su}(2^n)$ is simple and thus has no non-zero proper ideals, and so demonstrating that $\langle i\mathbb{P}\rangle_{\textrm{Lie}}$ {\em does} have a non-zero proper ideal suffices to show that $\langle i\mathbb{P}\rangle_{\textrm{Lie}} \neq \mathfrak{su}(2^n)$.

By definition, $i\mathbb{P}_{1} \subset \mathfrak{a}$ and $i\mathbb{P}_{1} \neq \emptyset$, and so $\mathfrak{a}$ is non-zero. To see that $\mathfrak{a} \subsetneq \langle i\mathbb{P}\rangle_{\textrm{Lie}}$, consider $iQ \in i\mathbb{P}_{2}$. The claim is that $iQ \notin \mathfrak{a}$, which we show by the following reasoning. Suppose that $iQ \in \mathfrak{a}$, meaning that 
\begin{align}
iQ = \sum_{iH \in \langle i\mathbb{P}_{1}\rangle_{[\cdot, \cdot]}} \alpha_{iH} iH,
\end{align}
where the coefficients $\alpha_{iH}$ are real and at least one is non-zero, and where each $iH \propto iP$ for some $iP \in i\mathcal{P}_{n}^{*}$. Since the Pauli strings form a basis, and since $iQ \in i\mathbb{P}_{2} \subseteq i\mathcal{P}_{n}^{*}$ is in particular an element of that basis, is must also be the case that at most one (and hence precisely one) of the $\alpha_{iH}$ is non-zero, as otherwise this would contradict the linear independence of the Pauli string basis. Furthermore, since each $iH \in \langle i\mathbb{P}_{1}\rangle_{[\cdot, \cdot]}$ is proportional to a Pauli string, the single non-zero $\alpha_{iH}$ must be $\pm 1$, and so we either have that $iQ \in i\mathbb{P}_{1}$, which contradicts the assumptions on $i\mathbb{P}_{1}$ and $i\mathbb{P}_{2}$, or there is some $r\ge 2$ and $iP_{1},\dots, iP_{r} \in i\mathbb{P}_{1}$ such that 
\begin{align}
    iQ = \textrm{ad}_{iP_{1}} \circ \dots \circ \textrm{ad}_{iP_{r-1}}(iP_{r}).
\end{align}
By \Cref{lem:last_anti_commute}, it must be that $\textrm{ad}_{iP_{1}}(iQ) \neq 0$, however this would entail that $\{iP_{1}, iQ \}$ is an edge in $G(i\mathbb{P})$, which contradicts the assumptions on disconnectedness made earlier. We thus have that $iQ \notin \mathfrak{a}$, and so $\mathfrak{a}$ is proper.

It remains, then, to show that $\mathfrak{a}$ is an ideal of $\langle i\mathbb{P}\rangle_{\textrm{Lie}}$, that is, for any $iA \in \mathfrak{a}$ and any $iB \in \langle i\mathbb{P}\rangle_{\textrm{Lie}}$, $[iA, iB] \in \mathfrak{a}$. Let us write
\begin{align}
    iB = \sum_{iC \in \langle i\mathbb{P}\rangle_{[\cdot, \cdot]}} \alpha_{iC}iC
\end{align}
with the $\alpha_{iC}$ being real and the terms $iC$ either being elements of $i\mathbb{P}$ or equal to $\textrm{ad}_{iR_{1}} \circ \dots \circ \textrm{ad}_{iR_{r-1}}(iR_{r})$ for some $r \ge 2$ and some $iR_{1}, \dots, iR_{r} \in i\mathbb{P}$. Since the elements of $i\mathbb{P}_{1}$ commute with all the elements of $i\mathbb{P}_{2}$, if the sequence $iR_{1}, \dots, iR_{r} \in i\mathbb{P}$ contains elements of both $i\mathbb{P}_{1}$ and $i\mathbb{P}_{2}$, then $\textrm{ad}_{iR_{1}} \circ \dots \circ \textrm{ad}_{iR_{r-1}}(iR_{r}) = 0$, by the following reasoning. Let $1 \le k < r$ be the largest value such that $iR_{k} \in i\mathbb{P}_{j}$ and $iR_{k+1} \in i\mathbb{P}_{1-j}$ for some $j \in \{0,1\}$. In particular, this means that $iR_{k+1}, \dots, iR_{r} \in i\mathbb{P}_{1-j}$. Since either 
\begin{align}
    \textrm{ad}_{iR_{k+1}} \circ \dots \circ \textrm{ad}_{iR_{r-1}}(iR_{r}) = 0
\end{align}
or 
\begin{align}
    \textrm{ad}_{iR_{k+1}} \circ \dots \circ \textrm{ad}_{iR_{r-1}}(iR_{r}) = \prod_{l=k+1}^{r}iR_{l},
\end{align}
it follows that either 
\begin{align}
    \textrm{ad}_{iR_{k}} \circ \textrm{ad}_{iR_{k+1}} \circ \dots \circ \textrm{ad}_{iR_{r-1}}(iR_{r}) = \textrm{ad}_{iR_{k}}(0)
\end{align}
or  
\begin{align}
    \textrm{ad}_{iR_{k}} \circ \textrm{ad}_{iR_{k+1}} \circ \dots \circ \textrm{ad}_{iR_{r-1}}(iR_{r}) = \textrm{ad}_{iR_{k}}\left(\prod_{l=k+1}^{r}iR_{l}\right).
\end{align}
In either case, $\textrm{ad}_{iR_{k}} \circ \textrm{ad}_{iR_{k+1}} \circ \dots \circ \textrm{ad}_{iR_{r-1}}(iR_{r}) = 0$, with the latter case following from the fact that $iR_{k}$ commutes with each of the $iR_{k+1}, \dots, iR_{r}$. Accordingly, we can rewrite $iB$ as
\begin{align}
    iB = \underbrace{\sum_{iC \in \langle i\mathbb{P}_{1}\rangle_{[\cdot, \cdot]}} \alpha_{iC}iC}_{:= iB_{1}} + \underbrace{\sum_{iD \in \langle i\mathbb{P}_{2}\rangle_{[\cdot, \cdot]}} \alpha_{iD}iD}_{:=iB_{2}}.
\end{align}
With this notation, we have that $[iA, iB] = [iA, iB_{1}] + [iA, iB_{2}]$. Since $iB_{1} \in \mathfrak{a}$, the properties of the Lie bracket ensure that $[iA, iB_{1}] \in \mathfrak{a}$. We conclude the proof by demonstrating that $[iA, iB_{2}]$ is necessarily $0$. 

The $iD \in \langle i\mathbb{P}_{2}\rangle_{[\cdot, \cdot]}$ for which $\alpha_{iD}iD$ is non-zero in the expansion of $iB_{2}$ can be written as $iD = \prod_{l = 1}^{t^{(D)}} iQ_{l}^{(D)}$ for some $t^{(D)} \ge 1$ and some $iQ_{1}^{(D)}, \dots, iQ_{t}^{(D)} \in i\mathbb{P}_{2}$. Similarly, if we write 
\begin{align}
iA:= \sum_{iA' \in \langle i\mathbb{P}_{1}\rangle_{[\cdot, \cdot]}} \beta_{iA'} iA',
\end{align}
then the $iA'$ corresponding to non-zero terms can be written as $iA' = \prod_{l = 1}^{r{(A')}} iP_{l}^{(A')}$ for some $r^{(A')} \ge 1$ and some $iP_{1}^{(A')}, \dots, iP_{r}^{(A')} \in i\mathbb{P}_{1}$. We thus have that
\begin{align}
    [iA, iB_{2}] = \sum_{\substack{iA' \in \langle i\mathbb{P}_{1}\rangle_{[\cdot, \cdot]}, \\ iD \in \langle i\mathbb{P}_{2}\rangle_{[\cdot, \cdot]}}} \alpha_{iD}\beta_{iA'} \left[ \prod_{l = 1}^{r^{(A')}} iP_{l}^{(A')}, \prod_{l = 1}^{t^{(D)}} iQ_{l}^{(D)}\right].
\end{align}
Since all the elements of $i\mathbb{P}_{1}$ commute with all the elements of $i\mathbb{P}_{2}$, we have that all the terms $ \left[ \prod_{l = 1}^{r^{(A')}} iP_{l}^{r^{(A')}}, \prod_{l = 1}^{t^{(D)}} iQ_{l}^{(D)}\right]$ are $0$. This demonstrates that $\mathfrak{a}$ is indeed an ideal of $\langle i\mathbb{P}\rangle_{\textrm{Lie}}$ and completes the proof.
\end{proof}

\section{Proof of Lemma~\ref{lem:isolation}}\label{app:proof_lem:isolation}

For completeness, let us restate the lemma here:

\begin{Lemma}[Pauli isolation]\label{lem:isolation_app}
Let $i\mathbb{Q}\subseteq i\mathcal{P}_n^*$ and let $iH \in \mathfrak{su}(2^n)$ be given by
\begin{align*}
iH=\sum_{iP\in i\mathcal{P}_n^*}\alpha_{iP}\, iP,
\end{align*}
with all coefficients $\alpha_{iP}$ being real. Let 
\begin{align}
    i\mathbb{P}(iH):=\{iP\in i\mathcal{P}_n^*\mid \alpha_{iP}\neq 0\}
\end{align}
and consider the anti-commutation graph 
\begin{align}
    G:=G\!\big(i\mathbb{Q}\cup i\mathbb{P}(iH)\big).
\end{align}
For $iP\in i\mathcal{P}_n^*$, if $iP\in \mathcal{V}_{G}(i\mathbb{Q})$, then $iP\in \langle i\mathbb{Q}\cup\{iH\}\rangle_{\mathrm{Lie}}$.
\end{Lemma}

\begin{proof} The anti-commutation graph $G$ has a vertex set which we identify with the set of Pauli strings $i\mathbb{Q}\cup i\mathbb{P}(iH)$. Recalling the definition of unique neighborhood from \Cref{sec:background}, a Pauli string $iP \in i\mathbb{P}(iH)$ has a unique neighborhood in $i\mathbb{Q}$ if $N_{G}(iP) \cap i\mathbb{Q} \neq 0$ and for all $iP' \in i\mathbb{P}(iH)$ such that $iP' \neq iP$, $N_{G}(iP') \cap i\mathbb{Q} \neq N_{G}(iP) \cap i\mathbb{Q}$. Let $N_{G}^{\textrm{uniq}}(i\mathbb{Q}) \subseteq i\mathbb{P}(iH)$ denote the set of all such elements with a unique neighborhood in $i\mathbb{Q}$ and recall that 
\begin{align}
    \mathcal{V}_{G}(i\mathbb{Q}):= i\mathbb{Q} \cup N_{G}^{\textrm{uniq}}(i\mathbb{Q}).
\end{align}
Since $i\mathbb{Q} \subset \langle i\mathbb{Q}\cup\{iH\}\rangle_{\mathrm{Lie}}$, if $iP \in i\mathbb{Q}$ it is directly seen to be an element of $\langle i\mathbb{Q}\cup\{iH\}\rangle_{\mathrm{Lie}}$ also, so let us consider only those $iP$ in $N_{G}^{\textrm{uniq}}(i\mathbb{Q})$ forthwith.

Let $iP \in N_{G}^{\textrm{uniq}}(i\mathbb{Q})$. The proof proceeds by demonstrating that there exists a sequence of elements $iQ_{1}, \dots, iQ_{r} \in N_{G}(iP) \cap i\mathbb{Q}$ such that 
\begin{align}
    \textrm{ad}_{iQ_{1}} \circ \dots \circ \textrm{ad}_{iQ_{r}}(iH) = \pm\alpha_{iP}iP.
\end{align}
Since $\textrm{ad}_{iQ_{1}} \circ \dots \circ \textrm{ad}_{iQ_{r}}(iH) \in \langle i\mathbb{Q}\cup\{iH\}\rangle_{[\cdot, \cdot]}$ and $\alpha_{iP}$ is real and non-zero, it then follows that $iP = (\pm \alpha_{iP})^{-1}\textrm{ad}_{iQ_{1}} \circ \dots \circ \textrm{ad}_{iQ_{r}}(iH) \in \langle i\mathbb{Q}\cup\{iH\}\rangle_{\mathrm{Lie}}$, as required.

Consider $iQ \in i\mathbb{Q}$. Noting that the elements of $i\mathbb{P}(iH)$ that anti-commute with $iQ$ are precisely the elements of $N_{G}(iQ) \cap i\mathbb{P}(iH)$, we get that
\begin{align}
    \textrm{ad}_{iQ}(iH) = \sum_{iP' \in N_{G}(iQ) \cap i\mathbb{P}(iH)} \alpha_{iP'}iQiP'.
\end{align}
Since $(iQ)^2 = i I \otimes \dots \otimes I$, applying $\textrm{ad}_{iQ}$ a second time produces
\begin{align}
    \textrm{ad}_{iQ}\circ \textrm{ad}_{iQ}(iH) = -\sum_{iP' \in N_{G}(iQ) \cap i\mathbb{P}(iH)} \alpha_{iP'}iP'.
\end{align}
Repeating this process for $iR \in i\mathbb{Q}$ such that $iR \neq iQ$, we get that
\begin{align}
    \textrm{ad}_{iR}\circ \textrm{ad}_{iR} &\circ \textrm{ad}_{iQ}\circ \textrm{ad}_{iQ}(iH) 
    = \sum_{iP' \in N_{G}(iR)\cap N_{G}(iQ) \cap i\mathbb{P}(iH)} \alpha_{iP'}iP'.
\end{align}
By the assumption that $iP$ has a unique neighborhood in $i\mathbb{Q}$, we get that
\begin{align}
    \left(\bigcap_{iQ \in N_{G}(iP) \cap i\mathbb{Q}} N_{G}(iQ)\right) \cap i\mathbb{P}(iH) = \{iP\}.
\end{align}
Writing $N_{G}(iP) \cap i\mathbb{Q} = \{iQ_{1}, \dots, iQ_{m}\}$, the above reasoning demonstrates that
\begin{align}
    \textrm{ad}_{iQ_{1}} \circ &\textrm{ad}_{iQ_{1}} \circ \textrm{ad}_{iQ_{2}} \circ \textrm{ad}_{iQ_{2}} 
    \dots \circ \textrm{ad}_{iQ_{m}}\circ \textrm{ad}_{iQ_{m}}(iH) = (-1)^{m}\alpha_{iP}iP,
\end{align}
completing the proof.
\end{proof}

\section{Proof of Theorem~\ref{thm:suf_ham}}\label{app:proof_thm:suf_ham}
For completeness, let us restate the theorem here:
\begin{Theorem}[Hamiltonian universality]\label{thm:suf_ham_app}
   Let $i\mathbb{Q}\subseteq i\mathcal{P}_n^*$ and let $iH \in \mathfrak{su}(2^n)$ be given by
\begin{align*}
iH=\sum_{iP\in i\mathcal{P}_n^*}\alpha_{iP}\, iP,
\end{align*}
with all coefficients $\alpha_{iP}$ being real. Let 
\begin{align}
    i\mathbb{P}(iH):=\{iP\in i\mathcal{P}_n^*\mid \alpha_{iP}\neq 0\}
\end{align}
and consider the anti-commutation graph 
\begin{align}
    G:=G\!\big(i\mathbb{Q}\cup i\mathbb{P}(iH)\big).
\end{align}

If, for some $r \in \mathbb{N}$, $\langle\mathcal{V}_G^{(r)}(i\mathbb{Q})\rangle_{\textrm{Lie}}=\mathfrak{su}(2^n)$ then $\langle i\mathbb{Q}\cup \{iH\}\rangle_{\textrm{Lie}}=\mathfrak{su}(2^n)$. 

If $\mathcal{V}_G^{(r)}(i\mathbb{Q}) = i\mathbb{Q}\cup i\mathbb{P}(iH)$, then the reverse implication also holds: if $\langle i\mathbb{Q}\cup \{iH\}\rangle_{\textrm{Lie}}=\mathfrak{su}(2^n)$ then $\langle\mathcal{V}_G^{(r)}(i\mathbb{Q})\rangle_{\textrm{Lie}}=\mathfrak{su}(2^n)$.
\end{Theorem}

\begin{proof} Suppose that, for some $r \in \mathbb{N}$, $\langle\mathcal{V}_G^{(r)}(i\mathbb{Q})\rangle_{\text{Lie}}=\mathfrak{su}(2^n)$. Since $\langle i\mathbb{Q}\cup \{iH\}\rangle_{\textrm{Lie}} \subseteq \mathfrak{su}(2^n)$, to demonstrate that $\langle i\mathbb{Q}\cup \{iH\}\rangle_{\textrm{Lie}}=\mathfrak{su}(2^n)$ it suffices to show that $\langle\mathcal{V}_G^{(r)}(i\mathbb{Q})\rangle_{\textrm{Lie}}  \subseteq \langle i\mathbb{Q}\cup \{iH\}\rangle_{\textrm{Lie}}$. We do so by showing that $\langle\mathcal{V}_G^{(r)}(i\mathbb{Q})\rangle_{\textrm{Lie}}  \subseteq \langle i\mathbb{Q}\cup \{iH\}\rangle_{\textrm{Lie}}$ in fact holds for all $r\in \mathbb{N}$, via an inductive argument.

By \Cref{lem:isolation}, we get that $\mathcal{V}_{G}(i\mathbb{Q}) \subset \langle i\mathbb{Q}\cup \{iH\}\rangle_{\textrm{Lie}}$, from which it follows that $\langle\mathcal{V}_G(i\mathbb{Q})\rangle_{\textrm{Lie}}  \subseteq \langle i\mathbb{Q}\cup \{iH\}\rangle_{\textrm{Lie}}$.

Now suppose that, for $t \in \mathbb{N}$, $\langle\mathcal{V}_G^{(t)}(i\mathbb{Q})\rangle_{\textrm{Lie}}  \subseteq \langle i\mathbb{Q}\cup \{iH\}\rangle_{\textrm{Lie}}$. Noting that $\mathcal{V}_G^{(t)}(i\mathbb{Q}) \subseteq i\mathbb{Q} \cup i\mathbb{P}(iH)$, it follows that $\mathcal{V}_G^{(t)}(i\mathbb{Q}) \cup i\mathbb{P}(iH)$ and hence also that the anti-commutation graph $G(\mathcal{V}_G^{(t)}(i\mathbb{Q}) \cup i\mathbb{P}(iH))$ is the same as the anti-commutation graph $G(i\mathbb{Q} \cup i\mathbb{P}(iH))$. Using \Cref{lem:isolation} again, we get that $\mathcal{V}_{G}(\mathcal{V}_G^{(t)}(i\mathbb{Q})) \subset \langle\mathcal{V}_G^{(t)}(i\mathbb{Q})\rangle_{\textrm{Lie}}$, which, by assumption, is contained in $\langle i\mathbb{Q}\cup \{iH\}\rangle_{\textrm{Lie}}$. Thus $\mathcal{V}_G^{(t+1)}(i\mathbb{Q}) \subset \langle i\mathbb{Q}\cup \{iH\}\rangle_{\textrm{Lie}}$, from which it follows that $\langle \mathcal{V}_G^{(t+1)}(i\mathbb{Q}) \rangle_{\textrm{Lie}} \subseteq \langle i\mathbb{Q}\cup \{iH\}\rangle_{\textrm{Lie}}$, completing the induction step and the first part of the proof.

Next, suppose that, for some $r\in \mathbb{N}$, $\mathcal{V}_G^{(r)}(i\mathbb{Q}) = i\mathbb{Q}\cup i\mathbb{P}(iH)$, and that $\langle i\mathbb{Q}\cup \{iH\}\rangle_{\textrm{Lie}}=\mathfrak{su}(2^n)$. By definition of $i\mathbb{P}(iH)$, $iH \in \textrm{span}_{\mathbb{R}}i\mathbb{P}(iH)$, meaning that $i\mathbb{Q} \cup \{iH\} \subset \langle i\mathbb{Q}\cup i\mathbb{P}(iH)\rangle_{\textrm{Lie}}$ and hence also that $\langle i\mathbb{Q} \cup \{iH\}\rangle_{\textrm{Lie}} \subset \langle i\mathbb{Q}\cup i\mathbb{P}(iH)\rangle_{\textrm{Lie}}$. Using the assumptions, we have that
\begin{align}
    \mathfrak{su}(2^n) &= \langle i\mathbb{Q}\cup \{iH\}\rangle_{\textrm{Lie}} \subseteq \langle i\mathbb{Q}\cup i\mathbb{P}(iH)\rangle_{\textrm{Lie}} = \langle \mathcal{V}_G^{(r)}(i\mathbb{Q})\rangle_{\textrm{Lie}},
\end{align}
and since $\langle \mathcal{V}_G^{(r)}(i\mathbb{Q})\rangle_{\textrm{Lie}} \subseteq \mathfrak{su}(2^{n})$ is guaranteed since $\mathcal{V}_G^{(r)}(i\mathbb{Q}) \subseteq i\mathcal{P}_{n}^*$, the result follows.
\end{proof}

\section{Proof of \Cref{cor:local_control}} \label{app:cor_local_control}

In this section, we prove \Cref{cor:local_control}, which we restate below for convenience. Beforehand, let us establish some further results used in the proof.

Let us define $i\mathbb{P}_{\text{s.q.}} \subseteq i\mathcal{P}_{n}^*$ the set of all weight $1$ Pauli strings with either a Pauli-$X$ or Pauli-$Z$ as the non-trivial factor, i.e.,
\begin{align}
i\mathbb{P}_{\text{s.q.}} := \left\{ iX_{1} \bigotimes_{j = 2}^n I_{j}, \  iZ_{1}  \bigotimes_{j =2}^2 I_{j}, \ iI_{1} \otimes X_{2}  \bigotimes_{j = 3}^n I_{j}, \ \dots, \ i \bigotimes_{j = 1}^{n-1} I_{j} \otimes X_{n}, \ i\bigotimes_{j = 1}^{n-1} I_{j} \otimes Z_{n}\right\}.
\end{align}
We have the following:
\begin{Lemma} \label{lem:uniq_nhbd_all_sq_control} If $iP, iQ \in i\mathcal{P}_{n}^*$ are such that $iP \neq iQ$, then 
\begin{align}
    \{iR \in i\mathbb{P}_{\text{s.q.}} | [iP, iR] \neq 0\} \neq \{iR \in i\mathbb{P}_{\text{s.q.}} | [iQ, iR] \neq 0\}.
\end{align}
\end{Lemma}

\begin{proof} Let $iP, iQ \in i\mathcal{P}_{n}^*$ be such that $iP \neq iQ$. Writing $iP = i\bigotimes_{j=1}^n p_{j}$ and $iQ = i\bigotimes_{j=1}^n q_{j}$ with the $p_{j},q_{k} \in \{I,X,Y,Z\}$, it must be the case that there exists some $k \in \{1,\dots, n\}$ such that $p_{k} \neq q_{k}$. If $[p_{k}, q_{k}] \neq 0$, let $r_{k} = p_{k}$ and define $iR = i r_{k} \bigotimes_{j \in [n]\setminus \{k\}} I_{j} \in i\mathbb{P}_{\text{s.q.}}$. We thus have that $[iP, iR] = 0$ but $[iQ, iR] = 0$, which proves the result in this case. If $[p_{k}, q_{k}] = 0$, it must be that either $p_{k} = I$ or $q_{k} = I$ (they cannot both be $I$ since $p_{k} \neq q_{k}$). Without loss of generality, suppose that $p_{k} = I$, in which case, let $r_{k} \in \{X,Y,Z\} \setminus\{q_{k}\}$. We thus have that $[p_{k}, r_{k}] = 0$ and $[q_{k}, r_{k}] \neq 0$, so defining $iR$ in an analogous manner to previously also proves the result in this case as well.
\end{proof}
Furthermore, the following was shown in Ref.~\cite{minimallyPQC} (see also Ref.~\cite{bremner2004fungible}):
\begin{Proposition} \label{prop:even_weight} Let $iP \in i\mathcal{P}_{k}^*$, $k \ge 2$, be such that $iP$ has full support, i.e., $\textrm{supp}(iP) = [k]$. Then $i\mathbb{P}_{\text{s.q.}} \cup \{iP\}$ is a universal generating set for $\mathfrak{su}(2^k)$ if and only if $k$ is even.
\end{Proposition}

We will also require the following:
\begin{Lemma} \label{lem:odd_support} Let $iP,iQ \in i\mathcal{P}_{n}^{*}$ be such that $[iP, iQ]\neq 0$ and $\textrm{supp}(iP)$ and $\textrm{supp}(iQ)$ are both odd. Then $\textrm{supp}([iP, iQ])$ is also odd.
\end{Lemma}
\begin{proof} Suppose that $iP,iQ \in i\mathcal{P}_{n}^{*}$ be such that $[iP, iQ]\neq 0$ and $\textrm{supp}(iP)$ and $\textrm{supp}(iQ)$ are both odd. Writing $iP := i\bigotimes_{j \in [n]} p_{j}$ and $iQ =i \bigotimes_{j \in [n]}q_{j}$ where $p_{j} = I$ if and only if $j \notin \textrm{supp}(iP)$ and $q_{j} = I$ if and only if $j \notin \textrm{supp}(iQ)$ by definition. We can write
\begin{align}
    [iP,iQ] = \bigotimes_{j \notin \textrm{supp}(iP) \cup \textrm{supp}(iQ)} I_{j} \bigotimes_{j \in \textrm{supp}(iP) \setminus \textrm{supp}(iQ)} p_{j}  \bigotimes_{j \in \textrm{supp}(iQ)\setminus \textrm{supp}(iP)}q_{j} \bigotimes_{j \in \textrm{supp}(iP) \cap \textrm{supp}(iQ)} p_{j}q_{j}.
\end{align}
The support of $[iP, iQ]$ is then given by 
\begin{align}
\textrm{supp}([iP, iQ]) = \textrm{supp}(iP) \setminus \textrm{supp}(iQ) \cup \textrm{supp}(iQ)\setminus \textrm{supp}(iP) \cup \{j \in  \textrm{supp}(iP) \cap \textrm{supp}(iQ): p_{j} \neq q_{j}\}.
\end{align}
Since $[iP, iQ] \neq 0$ by assumption, it must be the case that $|\{j \in  \textrm{supp}(iP) \cap \textrm{supp}(iQ): p_{j} \neq q_{j}\}|$ is odd. Since 
\begin{align}
|\textrm{supp}(iP) \setminus \textrm{supp}(iQ) \cup \textrm{supp}(iQ)\setminus \textrm{supp}(iP)| = |\textrm{supp}(iP)| + |\textrm{supp}(iQ)| - 2|\textrm{supp}(iP) \cap \textrm{supp}(iQ)|,
\end{align}
which is necessarily even (or zero), it follows that $|\textrm{supp}([iP, iQ])|$ must be odd.
\end{proof}

Let us now turn to the corollary: 
\begin{Corollary} \label{cor:local_control_app}  Let $i\mathbb{H}_{\text{s.q.}}\subseteq \mathfrak{su}(2^{n})$ be such that $\langle i\mathbb{H}_{\text{s.q.}} \rangle_{\textrm{Lie}} = \bigotimes_{j \in [n]}\mathfrak{su}(2)$. Let $iH\in\mathfrak{su}(2^n)$ and define $i\mathbb{H} := i\mathbb{H}_{\text{s.q.}} \cup \{iH\}$. The set $i\mathbb{H}$ is a universal generating set for $\mathfrak{su}(2^{n})$ if and only if the following hold:
\begin{enumerate}[label=(\arabic*)]
    \item \label{item:cor_lc_1} there exists an element $iP \in i\mathbb{P}(iH)$ such that $|\textrm{supp}(iP)|$ is even,
    \item \label{item:cor_lc_3} the anti-commutation graph $G(i\mathbb{P}(iH)\cup i\mathbb{P}_{\text{s.q.}})$ is connected.
\end{enumerate} 
\end{Corollary}

\begin{proof} Let us first show that $\langle i\mathbb{P}_{\text{s.q.}} \cup i\mathbb{P}(iH)\rangle_{\textrm{Lie}} = \langle i\mathbb{H} \rangle_{\textrm{Lie}}$. To show that $\langle i\mathbb{P}_{\text{s.q.}} \cup i\mathbb{P}(iH)\rangle_{\textrm{Lie}} \subseteq \langle i\mathbb{H} \rangle_{\textrm{Lie}}$, let us note that the assumption that $\langle i\mathbb{H}_{\textrm{s.q.}}\rangle_{\textrm{Lie}} = \bigotimes_{j=1}^{n} \mathfrak{su}(2)$ ensures that $i\mathbb{P}_{\text{s.q.}} \subseteq \langle i\mathbb{H}_{\textrm{s.q.}}\rangle_{\textrm{Lie}}$. Accordingly, we have that $i\mathbb{P}_{\text{s.q.}} \cup \{ iH\} \subseteq \langle i\mathbb{H}\rangle_{\textrm{Lie}}$. Let $G'$ denote the anti-commutation graph of $i\mathbb{P}_{\text{s.q.}} \cup i\mathbb{P}(iH)$. By \Cref{lem:uniq_nhbd_all_sq_control}, we have that every element of $i\mathbb{P}(iH)$ has a unique neighborhood in $i\mathbb{P}_{\text{s.q.}}$, meaning that $\mathcal{V}_{G'}(i\mathbb{P}_{\text{s.q.}}) = i\mathbb{P}_{\text{s.q.}} \cup i\mathbb{P}(iH)$. By \Cref{lem:isolation}, it follows that $i\mathbb{P}' := i\mathbb{P}_{\text{s.q.}} \cup i\mathbb{P}(iH) \subseteq \langle i\mathbb{H} \rangle_{\textrm{Lie}}$, from which $\langle i\mathbb{P}_{\text{s.q.}} \cup i\mathbb{P}(iH)\rangle_{\textrm{Lie}} \subseteq \langle i\mathbb{H} \rangle_{\textrm{Lie}}$ follows (recall that $\langle i\mathbb{H} \rangle_{\textrm{Lie}}$ is closed under taking Lie brackets and real linear combinations). 

To see that $\langle i\mathbb{H} \rangle_{\textrm{Lie}} \subseteq \langle i\mathbb{P}_{\text{s.q.}} \cup i\mathbb{P}(iH)\rangle_{\textrm{Lie}}$, it suffices to note that $\langle i\mathbb{P}_{\text{s.q.}}\rangle_{\textrm{Lie}} = \bigotimes_{j \in [n]}\mathfrak{su}(2)$, meaning that $i\mathbb{H}_{\textrm{s.q.}}\subset \langle i\mathbb{P}_{\text{s.q.}}\rangle_{\textrm{Lie}}$, and  that $iH\in \langle i\mathbb{P}(iH)\rangle_{\textrm{Lie}}$, allowing us to conclude that $i\mathbb{H}_{\text{s.q.}} \cup \{iH\} \subset \langle i\mathbb{P}_{\text{s.q.}} \cup i\mathbb{P}(iH)\rangle_{\textrm{Lie}}$. The fact that $\langle i\mathbb{H} \rangle_{\textrm{Lie}} \subseteq \langle i\mathbb{P}_{\text{s.q.}} \cup i\mathbb{P}(iH)\rangle_{\textrm{Lie}}$ then follows directly.

The above reasoning allows us to conclude that $\langle i\mathbb{H} \rangle_{\textrm{Lie}} = \mathfrak{su}(2^n)$ if and only if $\langle i\mathbb{P}_{\text{s.q.}} \cup i\mathbb{P}(iH)\rangle_{\textrm{Lie}} \subseteq \mathfrak{su}(2^n)$. The rest of the proof establishes that $\langle i\mathbb{P}_{\text{s.q.}} \cup i\mathbb{P}(iH)\rangle_{\textrm{Lie}} \subseteq \mathfrak{su}(2^n)$ if and only if 
\begin{enumerate}[label=(\arabic*)]
    \item \label{item:cor_lc_1_app} there exists an element $iP \in i\mathbb{P}(iH)$ such that $|\textrm{supp}(iP)|$ is even,
    \item \label{item:cor_lc_3_app} the anti-commutation graph $G(i\mathbb{P}(iH) \cup i\mathbb{P}_{\text{s.q.}})$ is connected,
\end{enumerate}
namely by appealing to \Cref{thm:nec_suf_pauli}. Since \Cref{item:cor_lc_3_app} is identical to \Cref{item:thm_3} of \Cref{thm:nec_suf_pauli}, we need only consider how \Cref{item:cor_lc_1_app} relates to \Cref{item:thm_1,item:thm_2} of the aforementioned theorem. Note that, since $i\mathbb{P}'$ contains $i\mathbb{P}_{\text{s.q.}}$, for any choice of subset $i\mathbb{Q} \subseteq i\mathbb{P}'$ for which $\textrm{supp}(i\mathbb{Q}) \subset [n]$, $i\mathbb{P}'_{|[n]\subseteq \textrm{supp}(i\mathbb{Q})}$ is product universal on $[n]\subseteq \textrm{supp}(i\mathbb{Q})$. Consequently, we in fact need only consider how \Cref{item:cor_lc_1_app} relates to \Cref{item:thm_1} of \Cref{thm:nec_suf_pauli}. By \Cref{prop:even_weight}, we know that, if $iP \in i\mathbb{P}(iH)$ is such that $|\textrm{supp}(iP)|$ is even, then $i\mathbb{Q} := \{iP\}\cup \{iQ \in i\mathbb{P}_{\text{s.q.}} : \textrm{supp}(iQ) \in \textrm{supp}(iP)\}$ is universal for $\mathfrak{su}(2^k)$ with $k = \textrm{supp}(i\mathbb{Q})$. This shows that if \Cref{item:cor_lc_1_app} holds, then \Cref{item:thm_1} of \Cref{thm:nec_suf_pauli} holds.

To complete the proof, we need to show the other implication also holds, which we do via the contrapositive, i.e., if \Cref{item:cor_lc_1_app} doesn't hold, then \Cref{item:thm_1} of \Cref{thm:nec_suf_pauli} also doesn't hold. Suppose that there does not exist an element of $i\mathbb{P}(iH)$ with even support. From \Cref{lem:odd_support}, we have that $\langle i\mathbb{P}'\rangle_{[\cdot, \cdot]}$ contains only elements of odd support and hence any choice of $i\mathbb{Q} \subseteq \langle i\mathbb{P}'\rangle_{[\cdot, \cdot]}$ does also. Fix such a choice of $i\mathbb{Q}$. For $i\mathbb{Q}$ to generate $\mathfrak{su}(2^k)$ for $k = |\textrm{supp}(i\mathbb{Q})|$, it must be the case that $\langle i\mathbb{Q}\rangle_{\textrm{Lie}}$ contains any universal generating set $i\mathbb{R}$ (recall \Cref{prop:univ_sets_in_Lie_closure_univ_sets}). If, in addition, $i\mathbb{R}\subseteq i\mathcal{P}_{n}^{*}$, then this is equivalent to requiring that $i\mathbb{R} \subseteq \langle i\mathbb{Q}\rangle_{[\cdot, \cdot]}$. Accordingly, by demonstrating that there exists a universal generating set $i\mathbb{R} \subseteq i\mathcal{P}_{n}^{*}$ such that $i\mathbb{R} \nsubseteq \langle i\mathbb{Q}\rangle_{[\cdot, \cdot]}$, we are finished.

To do so, let us order $\textrm{supp}(i\mathbb{Q})$ as $\{l_{1},...,l_{k}\}$ and let $iZ_{l_{a},l_{b}}$ be the Pauli string with Pauli-$Z$ operators at $l_{a}$ and $l_{b}$ and identities elsewhere. By taking 
\begin{align}
i\mathbb{R} := \{iQ \in i\mathbb{P}_{\text{s.q.}}: \textrm{supp}(iQ) \in \textrm{supp}(i\mathbb{Q})\} \cup \{iZ_{l_{j},l_{j+1}}: j \in [k-1]\}
\end{align}
By, e.g., Theorem $1$ in \cite{minimallyPQC}, $i\mathbb{R}$ is universal on $\textrm{supp}(i\mathbb{Q})$, however, by \Cref{lem:odd_support}, for all $j \in [k-1]$, $iZ_{l_{j},l_{j+1}} \notin  \langle i\mathbb{Q}\rangle_{[\cdot, \cdot]}$. We thus have that $i\mathbb{R} \nsubseteq \langle i\mathbb{Q}\rangle_{[\cdot, \cdot]}$ and hence that $i\mathbb{Q}$ itself cannot be universal. This shows that, in the case we are considering here, \Cref{item:thm_1} of \Cref{thm:nec_suf_pauli} holds if and only if $i\mathbb{P}(iH)$ contains an element with even support, completing the proof.
\end{proof}

\section{Proof of \Cref{cor:heisenberg}} \label{app:cor_heisenberg}
In this section, we prove \Cref{cor:heisenberg}, which we restate here for convenience.

\begin{Corollary}\label{cor:heisenberg_app}
    Let
    \begin{align*}
        iH=i\sum_{j=1}^{n-1} J_xX_jX_{j+1}+J_yY_jY_{j+1}+J_zZ_jZ_{j+1}
    \end{align*}
    for $J_x,J_y,J_z\in\mathbb{R}^*$ and let $i\mathbb{H}^{(k)}_{\text{s.q.}}\subseteq \mathfrak{su}(2^{n})$ be such that $\langle i\mathbb{H}^{(k)}_{\text{s.q.}} \rangle_{\textrm{Lie}} =\bigotimes_{j \in [k]}\mathfrak{su}(2)$ for some $k \in [n]$. Then, $\langle i\mathbb{H}^{(k)}_{\text{s.q.}} \cup \{iH\}\rangle_{\mathrm{Lie}}=\mathfrak{su}(2^n)$ if and only if $k \ge 2$.
\end{Corollary}

\begin{proof}
    The proof proceeds by showing (i) that $\langle i\mathbb{H}^{(k)}_{\text{s.q.}} \cup \{iH\}\rangle_{\mathrm{Lie}}  = \mathfrak{su}(2^n)$ if and only if $\langle i\mathbb{P}_{\text{s.q.}}^{(k)} \cup i\mathbb{P}(iH)\rangle_{\textrm{Lie}}  = \mathfrak{su}(2^n)$, and (ii) that $\langle i\mathbb{P}_{\text{s.q.}}^{(k)} \cup i\mathbb{P}(iH)\rangle_{\textrm{Lie}} = \mathfrak{su}(2^n)$ if and only if $k \ge 2$.

    For (i), the proof proceeds similarly to the first part of the proof of \Cref{cor:local_control_app}. Since every element of $i\mathbb{H}_{\textrm{s.q.}}^{(k)}$ is a linear combination of elements of $i\mathbb{P}_{\textrm{s.q.}}^{(k)}$, we have that $\langle i\mathbb{H}^{(k)}_{\text{s.q.}} \cup \{iH\}\rangle_{\mathrm{Lie}} \subseteq \langle i\mathbb{P}^{(k)}_{\text{s.q.}} \cup \{iH\}\rangle_{\mathrm{Lie}}$. By the assumption that $\langle i\mathbb{H}^{(k)}_{\text{s.q.}} \rangle_{\textrm{Lie}} =\bigotimes_{j \in [k]}\mathfrak{su}(2)$, we also have that $i\mathbb{P}_{\textrm{s.q.}}^{(k)} \subseteq \langle i\mathbb{H}^{(k)}_{\text{s.q.}} \rangle_{\textrm{Lie}}$, meaning that $\langle i\mathbb{P}^{(k)}_{\text{s.q.}} \cup \{iH\}\rangle_{\mathrm{Lie}} \subseteq \langle i\mathbb{H}^{(k)}_{\text{s.q.}} \cup \{iH\}\rangle_{\mathrm{Lie}}$ also, and thus that the two algebras are equal.

    Next, letting $iX_{1}$ denote the unit support Pauli string with a Pauli-$X$ in the first tensor factor and similarly for $iZ_{1}$, we get that $iX_{1},iZ_{1} \in i\mathbb{P}^{(k)}_{\text{s.q.}}$ for all $k \in [n]$. Let $G$ be anti-commutation graph $G:=G(\{iX_1,iZ_1\}\cup i\mathbb{P}(iH))$. We have that 
    \begin{align}
        \mathcal{V}_G^{(1)}(\{iX_1,iZ_1\})=\{iX_1,iZ_1,X_1X_2,Y_1Y_2,Z_1Z_2\}
    \end{align}
    and for each $l \in \{2,...,n-1\}$, 
    \begin{align}
         \mathcal{V}_G^{(l)}(\{iX_1,iZ_1\})&= \mathcal{V}_G^{(l-1)}(\{iX_1,iZ_1\})\cup \{iX_lX_{l+1},iY_lY_{l+1},iZ_lZ_{l+1}\}.
    \end{align}
    Accordingly, we have that $\mathcal{V}_G^{(n-1)}(\{iX_1,iZ_1\}) = \{iX_1,iZ_1\}\cup i\mathbb{P}(iH)$ and hence also that $\mathcal{V}_G^{(n-1)}(i\mathbb{P}_{\textrm{s.q.}}^{(k)}) = i\mathbb{P}_{\textrm{s.q.}}^{(k)}\cup i\mathbb{P}(iH)$ for all $k \in [n]$. By \Cref{thm:suf_ham}, we thus have that $i\mathbb{H}^{(k)}_{\text{s.q.}} \cup \{iH\}$ is universal if and only if $i\mathbb{P}_{\textrm{s.q.}}^{(k)}\cup i\mathbb{P}(iH)$ is universal.

    Since $\{iX_1,iZ_1\}\subseteq \langle i\mathbb{H}^{(k)}_{\textrm{s.q.}}\rangle_{\textrm{Lie}}$ for any $k\geq 1$ and $\mathcal{V}_G^{n-1}(\{iX_1,iZ_1\})=\{iX_1,iZ_1\}\cup i\mathbb{P}(iH)$, \Cref{thm:suf_ham} implies that $\langle i\mathbb{H}^{(k)}_{\text{s.q.}} \cup \{iH\}\rangle_{\mathrm{Lie}}=\mathfrak{su}(2^n)$ if and only if $\langle i\mathbb{H}^{(k)}_{\text{s.q.}} \cup i\mathbb{P}(iH)\rangle_{\mathrm{Lie}}=\mathfrak{su}(2^n)$. That is, we can rely on \Cref{thm:nec_suf_pauli} to prove universality.

    For (ii), we show that if $k\geq 2$, then $\langle i\mathbb{P}_{\text{s.q.}}^{(k)} \cup i\mathbb{P}(iH)\rangle_{\textrm{Lie}} = \mathfrak{su}(2^n)$, while if $k = 1$, $\langle i\mathbb{P}_{\text{s.q.}}^{(k)} \cup i\mathbb{P}(iH)\rangle_{\textrm{Lie}} \neq \mathfrak{su}(2^n)$. For the latter, first note that 
    \begin{align}
        \langle i\mathbb{H}^{(1)}_{\text{s.q.}} \cup i\mathbb{P}(iH)\rangle_{\mathrm{Lie}}\subseteq \langle \{iX_1, iZ_1\} \cup \{iX_jX_{j+1},iY_jY_{j+1},iZ_jZ_{j+1}\}_{j\in [n-1]} \rangle_{\mathrm{Lie}}.
    \end{align}
    That is, it is sufficient to show that 
    \begin{align}
        \langle \{iX_1, iZ_1\} \cup \{iX_jX_{j+1},iY_jY_{j+1},iZ_jZ_{j+1}\}_{j\in [n-1]} \rangle_{[\cdot,\cdot]}\neq i\mathcal{P}_n.
    \end{align}

    Next we note that 
    \begin{align}
        \langle \{iX_1, iZ_1\} \cup \{iX_jX_{j+1},iY_jY_{j+1},iZ_jZ_{j+1}\}_{j\in [n-1]} \rangle_{[\cdot,\cdot]}\subseteq \langle \{iX_1, iZ_1\} \cup \{iX_jX_{j+1},iZ_jZ_{j+1}\}_{j\in [n-1]} \rangle_{[\cdot,\cdot]}
    \end{align}
    by the following correspondence:
    \begin{equation}
    \begin{gathered}
        iY_1Y_2 = ad_{iZ_1}\circ ad_{iZ_1Z_2}\circ ad_{iZ_1}(iX_1X_2)\\
         iY_jY_{j+1} = ad_{iZ_{j-1}Z_j}\circ ad_{iX_{j-1}X_j}\circ ad_{iZ_jZ_{j+1}}\circ ad_{iY_{j-1}Y_j}(iX_jX_{j+1}), \quad j \in \{2,...,n-1\}.
    \end{gathered}
    \end{equation}
    Now we find that
    \begin{align}
        |\{iX_1, iZ_1\} \cup \{iX_jX_{j+1},iZ_jZ_{j+1}\}_{j\in [n-1]}| = 2n.
    \end{align}
    However, in Ref.~\cite{smith2025optimally} it was shown that any set of Pauli operators needs to contain at least $2n+1$ elements to be universal. That is, we have shown 
    \begin{align}
        \langle i\mathbb{H}^{(1)}_{\text{s.q.}} \cup i\mathbb{P}(iH)\rangle_{\mathrm{Lie}}\subseteq \langle \{iX_1, iZ_1\} \cup \{iX_jX_{j+1},iZ_jZ_{j+1}\}_{j\in [n-1]}\rangle_{\mathrm{Lie}}\neq \mathfrak{su}(2^n).
    \end{align}
    
    To show that if $k\geq 2$, then $\langle i\mathbb{P}_{\text{s.q.}}^{(k)} \cup i\mathbb{P}(iH)\rangle_{\textrm{Lie}} = \mathfrak{su}(2^n)$, it suffices to consider the case where $k = 2$, since for all $k \ge 2$, $\langle i\mathbb{P}_{\text{s.q.}}^{(2)} \cup i\mathbb{P}(iH)\rangle_{\textrm{Lie}} \subseteq \langle i\mathbb{P}_{\text{s.q.}}^{(k)} \cup i\mathbb{P}(iH)\rangle_{\textrm{Lie}}$. Let $i\mathbb{Q} := \{iX_{1}, iZ_{1}, iX_{2}, iZ_{2}, iZ_{1}Z_{2}\}$ which is a direct subset of $i\mathbb{P}_{\text{s.q.}}^{(2)} \cup i\mathbb{P}(iH)$ and hence is contained in $\langle i\mathbb{P}_{\text{s.q.}}^{(2)} \cup i\mathbb{P}(iH)\rangle_{[\cdot, \cdot]}$. Furthermore, $i\mathbb{Q}$ generates $\mathfrak{su}(4)$ on its support, and hence satisfies \Cref{item:thm_1} in \Cref{thm:nec_suf_pauli}. Next, the restriction of $i\mathbb{P}_{\text{s.q.}}^{(2)} \cup i\mathbb{P}(iH)$ to $[n]\setminus \textrm{supp}(i\mathbb{Q}) = \{3,...,n\}$ is $\{iX_{3}, iZ_{3}\}$ if $n = 3$ and
    \begin{align}
        \{iX_{3}, iZ_{3}, iX_{j}X_{j+1}, iZ_{j}Z_{j+1}\}_{j=3}^{n-1}
    \end{align}
otherwise. In either case, this set is product universal on $\{3,...,n\}$, so \Cref{item:thm_2} of \Cref{thm:nec_suf_pauli} is also satisfied. Finally, the anti-commutation graph of $i\mathbb{P}_{\text{s.q.}}^{(2)} \cup i\mathbb{P}(iH)$ is readily seen to be connected, meaning that all criteria of \Cref{thm:nec_suf_pauli} are satisfied, and thus $\langle i\mathbb{P}_{\text{s.q.}}^{(2)} \cup i\mathbb{P}(iH)\rangle_{\textrm{Lie}} = \mathfrak{su}(2^n)$. As outlined above, this completes the proof that $\langle i\mathbb{H}^{(k)}_{\text{s.q.}} \cup \{iH\}\rangle_{\mathrm{Lie}}  = \mathfrak{su}(2^n)$ if and only if $k \ge 2$.
\end{proof}

\section{Nontrivial, non-universal example}\label{app:counterexample}

As part of the discussion in \cref{sec:discussion} about generalizing Result~\ref{res:universality_H}, we consider some nontrivial counterexamples.
To corroborate this discussion, we prove the following statement:
\begin{Proposition}
    Let
    \begin{align}
    A:=iH_1^{(2)}&=\alpha\, iIX,\\
    B:=iH_2^{(2)}&=\beta_1\, iIY+\beta_2\, iXZ+\beta_3\, iYZ+\beta_4\, iZZ,
    \end{align}
    where $\alpha,\beta_1,\beta_2,\beta_3,\beta_4\in\mathbb{R}^*$.
    Then,
    \begin{align}
    \langle A,B\rangle_{\mathrm{Lie}}
    =
    \operatorname{span}_{\mathbb R}\{A,B,[A,B]\},
    \end{align}
    and consequently,
    \begin{align}
    \dim \langle A,B\rangle_{\mathrm{Lie}}=3.
    \end{align}
    In particular,
    \begin{align}
    \langle A,B\rangle_{\mathrm{Lie}}\subsetneq \mathfrak{su}(4).
    \end{align}
\end{Proposition}
\begin{proof}
    We will show that the real vector space
    \begin{align}
        L:= \mathrm{span}_\mathbb{R}(A,B,C),
    \end{align}
    where $C:=[A,B]$, is a 3-dimensional subalgebra of $\mathfrak{su}(4)$ containing $A$ and $B$. This will imply,
    \begin{align}
        \langle A,B\rangle_{\mathrm{Lie}} = L.
    \end{align}

    Consider the following operator,
    \begin{align}
        C&:= [A,B]\nonumber\\
        &=2AB\nonumber\\
        &= 2\alpha \left(-\beta_1 iIZ + \beta_2 iXY + \beta_3 i YY + \beta_4 iZY \right).
    \end{align}
    The second line can be obtained from realizing that $A$ anticommutes with $B$.
    We note further that each term in $B$ anticommutes with each other term in $B$, such that,
    \begin{align}
        B^2=-sI_4
    \end{align}
    where $s:=\beta_1^2+\beta_2^2+\beta_3^2+\beta_4^2>0$ and $I_4$ is the identity on all $4$ qubits. 
    Trivially, we also have that
    \begin{align}
        A^2 = -\alpha^2 I_4.
    \end{align}

    With the above we can show,
    \begin{align}
        [A,C] &= [A, 2AB] = 2(A^2B - ABA)\nonumber\\
        &= 2(-\alpha^2 B + A^2B)\nonumber\\
        &= -4\alpha^2 B \\
        [B,C] &= [B,2AB] = 2(BAB - AB^2)\nonumber\\
        &= 2(-B^2A + s A)\nonumber\\
        &= 4s A.
    \end{align}

    Since $[A,A]=[B,B]=[C,C]=0$, it follows from biliearity and antisymmetry of the commutator that $L$ is closed under commutator, i.e., $[L,L]\subseteq L$. Therefore, $L$ is a proper subalgebra of $\mathfrak{su}(4)$.

    Since $A,B\in L$ and $C\in \langle A,B\rangle_\mathrm{Lie}$, we find 
    \begin{align}
        L= \langle A,B\rangle_\mathrm{Lie}.
    \end{align}

    Since $A,B,C$ are linearly independent, we have $\dim(L)=3$.
    Since $\dim(\mathfrak{su}(4))=15$, it follows that
    \begin{align}
        \langle A,B\rangle_\mathrm{Lie} = L \subsetneq \mathfrak{su}(4). 
    \end{align}
    
\end{proof}
\end{document}